\documentclass[final]{IEEEtran}

\usepackage{graphicx}
\usepackage{subfigure}
\usepackage{multirow}
\usepackage{array}
\newcolumntype{P}[1]{>{\centering\arraybackslash}p{#1}}
\newcolumntype{M}[1]{>{\centering\arraybackslash}m{#1}}

\usepackage{amsthm,amssymb,amsmath,graphicx,multirow,color,amsfonts}%
\usepackage[update,prepend]{epstopdf}
\usepackage{multirow}
\usepackage[latin1]{inputenc}
\usepackage{tikz}
\usepackage{bbm} % for \mathbbm{1}
\usepackage{pdfpages}
\usepackage{multirow}
\usepackage{subfig}
\usepackage{comment}

\usepackage{graphicx}
\usepackage{multicol}
\usepackage{cite}

\usepackage[justification=centering]{caption}
\usepackage{textcomp}
\usepackage{psfrag}
\usepackage{arydshln}
\usepackage{url}
\usepackage{soul}
\usepackage{graphicx,color}
\usepackage[nolist]{acronym}
\usepackage{algorithm,algorithmic} %algorithm
\usepackage{mathtools,lipsum}
\usepackage{cuted}
\usepackage{amsmath}

\usepackage{mathrsfs}

\usepackage[capitalise]{cleveref}
\Crefname{equation}{Eq.\!}{Eqs.\!}
\Crefname{figure}{Fig.\!}{Figs.\!}
\Crefname{tabular}{Tab.\!}{Tabs.\!}
\Crefname{section}{Section\!}{Sections.\!}

\def\nb0{{\mathbf{0}}}
\def\nb1{{\mathbf{1}}}

\newtheorem{lemma}{Lemma}

\newtheorem{definition}{Definition}

\newtheorem{theorem}{Theorem}

\newtheorem{corollary}{Corollary}

\newenvironment{sequation}{
\begin{equation}\small}{\end{equation}
}

\begin{document}
%\pagenumbering{gobble}
\graphicspath{{./Figures/}}
	\begin{acronym}

\acro{5G-NR}{5G New Radio}
\acro{3GPP}{3rd Generation Partnership Project}
\acro{ABS}{aerial base station}
\acro{AC}{address coding}
\acro{ACF}{autocorrelation function}
\acro{ACR}{autocorrelation receiver}
\acro{ADC}{analog-to-digital converter}
\acrodef{aic}[AIC]{Analog-to-Information Converter}     
\acro{AIC}[AIC]{Akaike information criterion}
\acro{aric}[ARIC]{asymmetric restricted isometry constant}
\acro{arip}[ARIP]{asymmetric restricted isometry property}

\acro{ARQ}{Automatic Repeat Request}
\acro{AUB}{asymptotic union bound}
\acrodef{awgn}[AWGN]{Additive White Gaussian Noise}     
\acro{AWGN}{additive white Gaussian noise}

\acro{APSK}[PSK]{asymmetric PSK} 

\acro{waric}[AWRICs]{asymmetric weak restricted isometry constants}
\acro{warip}[AWRIP]{asymmetric weak restricted isometry property}
\acro{BCH}{Bose, Chaudhuri, and Hocquenghem}        
\acro{BCHC}[BCHSC]{BCH based source coding}
\acro{BEP}{bit error probability}
\acro{BFC}{block fading channel}
\acro{BG}[BG]{Bernoulli-Gaussian}
\acro{BGG}{Bernoulli-Generalized Gaussian}
\acro{BPAM}{binary pulse amplitude modulation}
\acro{BPDN}{Basis Pursuit Denoising}
\acro{BPPM}{binary pulse position modulation}
\acro{BPSK}{Binary Phase Shift Keying}
\acro{BPZF}{bandpass zonal filter}
\acro{BSC}{binary symmetric channels}              
\acro{BU}[BU]{Bernoulli-uniform}
\acro{BER}{bit error rate}
\acro{BS}{base station}
\acro{BW}{BandWidth}
\acro{BLLL}{ binary log-linear learning }

\acro{CP}{Cyclic Prefix}
\acrodef{cdf}[CDF]{cumulative distribution function}   
\acro{CDF}{Cumulative Distribution Function}
\acrodef{c.d.f.}[CDF]{cumulative distribution function}
\acro{CCDF}{complementary cumulative distribution function}
\acrodef{ccdf}[CCDF]{complementary CDF}               
\acrodef{c.c.d.f.}[CCDF]{complementary cumulative distribution function}
\acro{CD}{cooperative diversity}

\acro{CDMA}{Code Division Multiple Access}
\acro{ch.f.}{characteristic function}
\acro{CIR}{channel impulse response}
\acro{cosamp}[CoSaMP]{compressive sampling matching pursuit}
\acro{CR}{cognitive radio}
\acro{cs}[CS]{compressed sensing}                   
\acrodef{cscapital}[CS]{Compressed sensing} %will not include it in the list
\acrodef{CS}[CS]{compressed sensing}
\acro{CSI}{channel state information}
\acro{CCSDS}{consultative committee for space data systems}
\acro{CC}{convolutional coding}
\acro{Covid19}[COVID-19]{Coronavirus disease}

\acro{DAA}{detect and avoid}
\acro{DAB}{digital audio broadcasting}
\acro{DCT}{discrete cosine transform}
\acro{dft}[DFT]{discrete Fourier transform}
\acro{DR}{distortion-rate}
\acro{DS}{direct sequence}
\acro{DS-SS}{direct-sequence spread-spectrum}
\acro{DTR}{differential transmitted-reference}
\acro{DVB-H}{digital video broadcasting\,--\,handheld}
\acro{DVB-T}{digital video broadcasting\,--\,terrestrial}
\acro{DL}{DownLink}
\acro{DSSS}{Direct Sequence Spread Spectrum}
\acro{DFT-s-OFDM}{Discrete Fourier Transform-spread-Orthogonal Frequency Division Multiplexing}
\acro{DAS}{Distributed Antenna System}
\acro{DNA}{DeoxyriboNucleic Acid}

\acro{EC}{European Commission}
\acro{EED}[EED]{exact eigenvalues distribution}
\acro{EIRP}{Equivalent Isotropically Radiated Power}
\acro{ELP}{equivalent low-pass}
\acro{eMBB}{Enhanced Mobile Broadband}
\acro{EMF}{ElectroMagnetic Field}
\acro{EU}{European union}
\acro{EI}{Exposure Index}
\acro{eICIC}{enhanced Inter-Cell Interference Coordination}

\acro{FC}[FC]{fusion center}
\acro{FCC}{Federal Communications Commission}
\acro{FEC}{forward error correction}
\acro{FFT}{fast Fourier transform}
\acro{FH}{frequency-hopping}
\acro{FH-SS}{frequency-hopping spread-spectrum}
\acrodef{FS}{Frame synchronization}
\acro{FSsmall}[FS]{frame synchronization}  
\acro{FDMA}{Frequency Division Multiple Access}

\acro{GA}{Gaussian approximation}
\acro{GF}{Galois field }
\acro{GG}{Generalized-Gaussian}
\acro{GIC}[GIC]{generalized information criterion}
\acro{GLRT}{generalized likelihood ratio test}
\acro{GPS}{Global Positioning System}
\acro{GMSK}{Gaussian Minimum Shift Keying}
\acro{GSMA}{Global System for Mobile communications Association}
\acro{GS}{ground station}
\acro{GMG}{ Grid-connected MicroGeneration}

\acro{HAP}{high altitude platform}
\acro{HetNet}{Heterogeneous network}

\acro{IDR}{information distortion-rate}
\acro{IFFT}{inverse fast Fourier transform}
\acro{iht}[IHT]{iterative hard thresholding}
\acro{i.i.d.}{independent, identically distributed}
\acro{IoT}{Internet of Things}                      
\acro{IR}{impulse radio}
\acro{lric}[LRIC]{lower restricted isometry constant}
\acro{lrict}[LRICt]{lower restricted isometry constant threshold}
\acro{ISI}{intersymbol interference}
\acro{ITU}{International Telecommunication Union}
\acro{ICNIRP}{International Commission on Non-Ionizing Radiation Protection}
\acro{IEEE}{Institute of Electrical and Electronics Engineers}
\acro{ICES}{IEEE international committee on electromagnetic safety}
\acro{IEC}{International Electrotechnical Commission}
\acro{IARC}{International Agency on Research on Cancer}
\acro{IS-95}{Interim Standard 95}

\acro{KPI}{Key Performance Indicator}

\acro{LEO}{low earth orbit}
\acro{LF}{likelihood function}
\acro{LLF}{log-likelihood function}
\acro{LLR}{log-likelihood ratio}
\acro{LLRT}{log-likelihood ratio test}
\acro{LoS}{Line-of-Sight}
\acro{LRT}{likelihood ratio test}
\acro{wlric}[LWRIC]{lower weak restricted isometry constant}
\acro{wlrict}[LWRICt]{LWRIC threshold}
\acro{LPWAN}{Low Power Wide Area Network}
\acro{LoRaWAN}{Low power long Range Wide Area Network}
\acro{NLoS}{Non-Line-of-Sight}
\acro{LiFi}[Li-Fi]{light-fidelity}
 \acro{LED}{light emitting diode}
 \acro{LABS}{LoS transmission with each ABS}
 \acro{NLABS}{NLoS transmission with each ABS}

\acro{MB}{multiband}
\acro{MC}{macro cell}
\acro{MDS}{mixed distributed source}
\acro{MF}{matched filter}
\acro{m.g.f.}{moment generating function}
\acro{MI}{mutual information}
\acro{MIMO}{Multiple-Input Multiple-Output}
\acro{MISO}{multiple-input single-output}
\acrodef{maxs}[MJSO]{maximum joint support cardinality}                       
\acro{ML}[ML]{maximum likelihood}
\acro{MMSE}{minimum mean-square error}
\acro{MMV}{multiple measurement vectors}
\acrodef{MOS}{model order selection}
\acro{M-PSK}[${M}$-PSK]{$M$-ary phase shift keying}                       
\acro{M-APSK}[${M}$-PSK]{$M$-ary asymmetric PSK} 
\acro{MP}{ multi-period}
\acro{MINLP}{mixed integer non-linear programming}

\acro{M-QAM}[$M$-QAM]{$M$-ary quadrature amplitude modulation}
\acro{MRC}{maximal ratio combiner}                  
\acro{maxs}[MSO]{maximum sparsity order}                                      
\acro{M2M}{Machine-to-Machine}                                                
\acro{MUI}{multi-user interference}
\acro{mMTC}{massive Machine Type Communications}      
\acro{mm-Wave}{millimeter-wave}
\acro{MP}{mobile phone}
\acro{MPE}{maximum permissible exposure}
\acro{MAC}{media access control}
\acro{NB}{narrowband}
\acro{NBI}{narrowband interference}
\acro{NLA}{nonlinear sparse approximation}
\acro{NLOS}{Non-Line of Sight}
\acro{NTIA}{National Telecommunications and Information Administration}
\acro{NTP}{National Toxicology Program}
\acro{NHS}{National Health Service}

\acro{LOS}{Line of Sight}

\acro{OC}{optimum combining}                             
\acro{OC}{optimum combining}
\acro{ODE}{operational distortion-energy}
\acro{ODR}{operational distortion-rate}
\acro{OFDM}{Orthogonal Frequency-Division Multiplexing}
\acro{omp}[OMP]{orthogonal matching pursuit}
\acro{OSMP}[OSMP]{orthogonal subspace matching pursuit}
\acro{OQAM}{offset quadrature amplitude modulation}
\acro{OQPSK}{offset QPSK}
\acro{OFDMA}{Orthogonal Frequency-division Multiple Access}
\acro{OPEX}{Operating Expenditures}
\acro{OQPSK/PM}{OQPSK with phase modulation}

\acro{PAM}{pulse amplitude modulation}
\acro{PAR}{peak-to-average ratio}
\acrodef{pdf}[PDF]{probability density function}                      
\acro{PDF}{probability density function}
\acrodef{p.d.f.}[PDF]{probability distribution function}
\acro{PDP}{power dispersion profile}
\acro{PMF}{probability mass function}                             
\acrodef{p.m.f.}[PMF]{probability mass function}
\acro{PN}{pseudo-noise}
\acro{PPM}{pulse position modulation}
\acro{PRake}{Partial Rake}
\acro{PSD}{power spectral density}
\acro{PSEP}{pairwise synchronization error probability}
\acro{PSK}{phase shift keying}
\acro{PD}{power density}
\acro{8-PSK}[$8$-PSK]{$8$-phase shift keying}
\acro{PPP}{Poisson point process}
\acro{PCP}{Poisson cluster process}
 
\acro{FSK}{Frequency Shift Keying}

\acro{QAM}{Quadrature Amplitude Modulation}
\acro{QPSK}{Quadrature Phase Shift Keying}
\acro{OQPSK/PM}{OQPSK with phase modulator }

\acro{RD}[RD]{raw data}
\acro{RDL}{"random data limit"}
\acro{ric}[RIC]{restricted isometry constant}
\acro{rict}[RICt]{restricted isometry constant threshold}
\acro{rip}[RIP]{restricted isometry property}
\acro{ROC}{receiver operating characteristic}
\acro{rq}[RQ]{Raleigh quotient}
\acro{RS}[RS]{Reed-Solomon}
\acro{RSC}[RSSC]{RS based source coding}
\acro{r.v.}{random variable}                               
\acro{R.V.}{random vector}
\acro{RMS}{root mean square}
\acro{RFR}{radiofrequency radiation}
\acro{RIS}{Reconfigurable Intelligent Surface}
\acro{RNA}{RiboNucleic Acid}
\acro{RRM}{Radio Resource Management}
\acro{RUE}{reference user equipments}
\acro{RAT}{radio access technology}
\acro{RB}{resource block}

\acro{SA}[SA-Music]{subspace-augmented MUSIC with OSMP}
\acro{SC}{small cell}
\acro{SCBSES}[SCBSES]{Source Compression Based Syndrome Encoding Scheme}
\acro{SCM}{sample covariance matrix}
\acro{SEP}{symbol error probability}
\acro{SG}[SG]{sparse-land Gaussian model}
\acro{SIMO}{single-input multiple-output}
\acro{SINR}{signal-to-interference plus noise ratio}
\acro{SIR}{signal-to-interference ratio}
\acro{SISO}{Single-Input Single-Output}
\acro{SMV}{single measurement vector}
\acro{SNR}[\textrm{SNR}]{signal-to-noise ratio} 
\acro{sp}[SP]{subspace pursuit}
\acro{SS}{spread spectrum}
\acro{SW}{sync word}
\acro{SAR}{specific absorption rate}
\acro{SSB}{synchronization signal block}
\acro{SR}{shrink and realign}

\acro{tUAV}{tethered Unmanned Aerial Vehicle}
\acro{TBS}{terrestrial base station}

\acro{uUAV}{untethered Unmanned Aerial Vehicle}
\acro{PDF}{probability density functions}

\acro{PL}{path-loss}

\acro{TH}{time-hopping}
\acro{ToA}{time-of-arrival}
\acro{TR}{transmitted-reference}
\acro{TW}{Tracy-Widom}
\acro{TWDT}{TW Distribution Tail}
\acro{TCM}{trellis coded modulation}
\acro{TDD}{Time-Division Duplexing}
\acro{TDMA}{Time Division Multiple Access}
\acro{Tx}{average transmit}

\acro{UAV}{Unmanned Aerial Vehicle}
\acro{uric}[URIC]{upper restricted isometry constant}
\acro{urict}[URICt]{upper restricted isometry constant threshold}
\acro{UWB}{ultrawide band}
\acro{UWBcap}[UWB]{Ultrawide band}   
\acro{URLLC}{Ultra Reliable Low Latency Communications}
         
\acro{wuric}[UWRIC]{upper weak restricted isometry constant}
\acro{wurict}[UWRICt]{UWRIC threshold}                
\acro{UE}{User Equipment}
\acro{UL}{UpLink}

\acro{WiM}[WiM]{weigh-in-motion}
\acro{WLAN}{wireless local area network}
\acro{wm}[WM]{Wishart matrix}                               
\acroplural{wm}[WM]{Wishart matrices}
\acro{WMAN}{wireless metropolitan area network}
\acro{WPAN}{wireless personal area network}
\acro{wric}[WRIC]{weak restricted isometry constant}
\acro{wrict}[WRICt]{weak restricted isometry constant thresholds}
\acro{wrip}[WRIP]{weak restricted isometry property}
\acro{WSN}{wireless sensor network}                        
\acro{WSS}{Wide-Sense Stationary}
\acro{WHO}{World Health Organization}
\acro{Wi-Fi}{Wireless Fidelity}

\acro{sss}[SpaSoSEnc]{sparse source syndrome encoding}

\acro{VLC}{Visible Light Communication}
\acro{VPN}{Virtual Private Network} 
\acro{RF}{Radio Frequency}
\acro{FSO}{Free Space Optics}
\acro{IoST}{Internet of Space Things}

\acro{GSM}{Global System for Mobile Communications}
\acro{2G}{Second-generation cellular network}
\acro{3G}{Third-generation cellular network}
\acro{4G}{Fourth-generation cellular network}
\acro{5G}{Fifth-generation cellular network}	
\acro{gNB}{next-generation Node-B Base Station}
\acro{NR}{New Radio}
\acro{UMTS}{Universal Mobile Telecommunications Service}
\acro{LTE}{Long Term Evolution}

\acro{QoS}{Quality of Service}
\end{acronym}
	
	%% EMF definitions
\newcommand{\SAR} {\mathrm{SAR}}
\newcommand{\WBSAR} {\mathrm{SAR}_{\mathsf{WB}}}
\newcommand{\gSAR} {\mathrm{SAR}_{10\si{\gram}}}
\newcommand{\Sab} {S_{\mathsf{ab}}}
\newcommand{\Eavg} {E_{\mathsf{avg}}}
\newcommand{\ft}{f_{\textsf{th}}}
\newcommand{\alphatf}{\alpha_{24}}

\title{
Coverage and Rate Analysis of Follower-Based LEO Satellite Networks: A Stochastic \\ Geometry Approach
}

\author{
Juanjuan Ru, Ruibo Wang, {\em Member, IEEE}, and Mohamed-Slim Alouini, {\em Fellow, IEEE}
\thanks{This work was supported by the King Abdullah University of Science and Technology Office of Sponsored Research.
\par
The authors are with King Abdullah University of Science and Technology, CEMSE division, Thuwal 23955-6900, Saudi Arabia. Corresponding author: Ruibo Wang. (e-mail: Juanjuan.ru@kaust.edu.sa; ruibo.wang@kaust.edu.sa; slim.alouini@kaust.edu.sa). 
}
\vspace{-6mm}
}
\maketitle

\begin{abstract}
To mitigate inter-satellite interference and payload limits in LEO mega-constellations, satellite clusters, groups of small cooperative satellites have been proposed to improve performance and reduce interference. The typical configuration divides the cluster into a leader satellite with full processing and control capabilities and multiple simpler follower satellites that assist with coverage and throughput. These clusters enhance coverage and throughput, prompting interest in their performance gains and optimal deployment. Given that the spherical stochastic geometry (SG) model has been proven effective for modeling such structures, we establish a performance evaluation framework based on the SG approach for the leader-follower satellite architecture, enabling an assessment of communication performance under different deployment configurations quantitatively. We derive analytical expressions for the outage probability and average data rate to evaluate the communication performance of the satellite system, along with low-complexity approximations. Numerical results demonstrate the performance advantages of the leader-follower architecture over a single leader satellite and explore optimal deployment configurations for the follower satellites. 
\end{abstract}

\begin{IEEEkeywords}
Satellite follower, stochastic geometry, performance analysis, binomial point process, LEO satellite network.
\end{IEEEkeywords}

\section{Introduction}
\subsection{Motivation}
In recent years, the rapid expansion of low Earth orbit (LEO) satellite constellations has enabled low-latency and seamless global coverage \cite{xiao2022leo}. Compared to terrestrial systems, satellite networks offer enhanced resilience and are particularly well suited for ultra-long-distance transmissions \cite{wang2022stochastic}. However, the construction of mega-constellations has also raised many issues. Firstly, for a constellation that is already large in scale, adding more satellites to the orbit may increase interference among the satellites \cite{dwivedi2023performance}. Secondly, the scalability of a single satellite's payload and service capabilities is limited by its inherent hardware. Inspired by this, the concept of satellite clusters has been proposed \cite{campbell2003planning}. A group of closely positioned small satellites can further reduce interference and enhance communication performance through cooperation \cite{liu2018survey}.

\par
Specifically, satellite clusters can be categorized into two types based on hardware configuration. The first type divides the cluster into "leader" and "follower" roles \cite{goh2019leader,ahn2012leader}. The leader satellite is responsible for managing the entire cluster and is equipped with a powerful onboard processor to support inter-satellite communication and routing. The follower satellites are designed to execute the leader's instructions and enhance the cluster's performance. Therefore, their functionality can be simplified as much as possible. For example, they may only be used for coverage extension and communication throughput enhancement, thus eliminating the need for a comprehensive protocol stack \cite{wang2025modeling}.

\par
The second type of satellite cluster is temporarily formed by a group of geographically proximate satellites within the constellation \cite{lee2024analyzing}. This type of cluster does not incur the additional cost of launching follower satellites, nor does it risk complete cluster failure due to a leader malfunction. However, since followers and leaders typically have similar hardware configurations, this structure may lead to wasted resources for followers. Moreover, as no additional follower satellites are introduced, the performance improvement of such clusters is limited. Therefore, we primarily focus on the first type of satellite cluster, that is, follower-based satellite clusters.

\subsection{Related Works}
To date, several studies have explored the feasibility, practical advantages, mathematical modeling approaches, and performance evaluation of the leader-follower architecture in satellite clusters. The Canadian nano-satellite program successfully demonstrated satellite clustering through minor orbital configuration adjustments \cite{eyer2007formation}. When observed from Earth, the follower satellite orbited the leader satellite in a circular path. The authors in \cite{popov2021development} further noted that by expanding to multiple satellites, formations consisting of more than two follower-leader satellites can be achieved, enabling coordinated group movement. Since the followers' orbits are determined by the leader's orbit, adding more followers does not require additional orbital resources \cite{yu2016virtual}. Moreover, the relative distance between satellites within the cluster ranges from just a few hundred meters to several hundred kilometers. Due to this geographic proximity, beam alignment, considered one of the primary challenges in inter-satellite links \cite{wang2024ultra}, can be achieved with relatively high accuracy.

\par
The authors in \cite{jung2023satellite} proposed two mathematical modeling approaches for the satellite cluster architecture. The first is termed the circular cluster, where follower satellites are equidistant from the leader satellite. The second is called the uniform cluster, where the leader resides at the center of a spherical cap while followers are uniformly distributed across its surface. Schematic diagrams of two types of follower-based satellite clusters are shown in Fig.~\ref{figure1}. Since uniform distribution maximizes spatial efficiency, under identical minimum distance constraints, a uniform formation can accommodate more followers than a circular formation. Therefore, this paper focuses on uniform clusters. The authors in \cite{jung2023satellite2} further demonstrated that the homogeneous binomial point process (BPP) can accurately model the spatial distribution of uniform clusters, validating BPP's applicability for modeling leader-follower satellite systems. Furthermore, studies have shown that the leader-follower architecture can enhance the performance of satellite communication networks through spatial diversity. Since the probability of establishing line-of-sight (LoS) links varies with elevation angle \cite{al2020modeling}, the low orbital altitude of LEO satellites results in a non-negligible probability of link blockage between satellites and ground terminals \cite{chae2023performance}. With the help of the spatial diversity of satellite followers, the probability that all satellite-to-terminal links within a cluster are non-line-of-sight simultaneously becomes much lower.
\begin{figure}[ht]
\centering
\includegraphics[width=\linewidth]{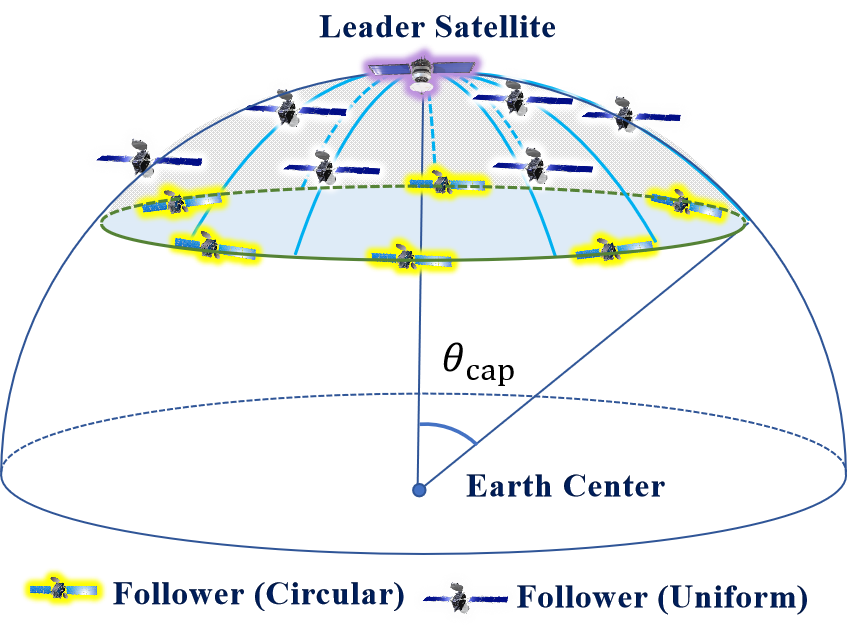}
\caption{Schematic diagram of two types of satellite clusters.}
\label{figure1}
\end{figure}

\par
So far, the study most relevant to this paper is \cite{jung2023satellite2}, which laid the foundation for the rationality and practical significance of the leader-follower structure. However, as a non-technical review, it did not establish a theoretical framework for performance analysis. In fact, a quantitative analytical framework is necessary to evaluate the improvement in communication performance brought by this structure. Through performance evaluation, we can quantify the enhancement of communication performance achieved by introducing followers and determine the optimal configuration of followers. Currently, the communication performance improvement of the leader-follower structure is limited to its spatial diversity. We aim to explore more possibilities to demonstrate the advantages of this structure. Since inter-satellite links operate in free space with short communication distances, their capacity is greater than that of satellite-to-ground links \cite{wang2025satellite}. Therefore, the leader can simultaneously transmit signals to multiple followers, using them as relays to forward the leader's communication data to the ground in parallel. Many LEO constellations, such as Starlink, already have the capability to establish multiple ISLs simultaneously \cite{chen20243}. In our view, the aggregate data rate of multiple follower satellites has the potential to exceed that of a single leader satellite.

\subsection{Contribution}
To our knowledge, this paper presents the first technical study on satellite followers based on stochastic geometry. The specific contributions are as follows:
\begin{itemize}
\item We establish a performance analysis framework for leader-follower satellite systems using the BPP model and the stochastic geometry tool, analyzing network outage probability and average data rate performance. Furthermore, we derive upper and lower bounds to provide low-complexity performance approximations.
\item We verify the accuracy of our analytical expressions through Monte Carlo simulations. As an application of our analytical framework, we study the impact of factors such as altitude, transmit power, and the number of follower satellites in the cluster on network performance.
\item We conduct case studies and quantitatively evaluate the coverage performance improvement of a leader-follower cluster compared to a single leader satellite. Additionally, we examine the average data rate enhancement achieved by employing multiple followers as relays versus a single leader. Based on these analyses, we validate both the spatial diversity advantages of satellite clusters and the effectiveness of follower structures in extending individual satellites' service capabilities. 
\end{itemize}

\section{System Model}
This section introduces the spatial distribution model of the follower-based LEO satellite network and the channel model of satellite-user links and inter-satellite links. Table~\ref{table1} provides the main symbols and their definition in this paper.

\begin{table*}[ht]
\centering
\caption{Simulation parameters \cite{wang2022conditional,jung2023satellite2,talgat2024maximizing}.}
\label{table1}
\renewcommand{\arraystretch}{1.1}
\resizebox{0.75\linewidth}{!}{
\begin{tabular}{|c|c|c|}
\hline
Notation   &  Parameter  & Default Value    \\ \hline \hline
$h_{\mathrm{Sat}}$   & Altitude of satellite  & $600$~km    \\ \hline
$R_{\oplus}$   & Earth radius  & $6371$~km    \\ \hline
$N_L$   & Number of leader satellites   & $1000$  \\ \hline
$N_F$  & Number of followers per leader satellite& $10$  \\ \hline
$\theta_{\mathrm{cap}}$   & Central angle of follower distribution cap & $1^\circ$  \\ \hline
$\theta_{\mathrm{max}}$   & Maximum angle between leader and user  & $\pi/4$  \\ \hline
$\gamma_{\mathrm{th}}$   & Coverage threshold  & $-5$~dB  \\ \hline
$\rho_{\mathrm{LU}}$, $\rho_{\mathrm{FU}}$, $\rho_{\mathrm{LF}}$   & Transmit power & $20$~dBW, $15$~dBW, $5$~dBW    \\ \hline
$G_{\mathrm{LU}}$, $G_{\mathrm{FU}}$, $G_{\mathrm{LF}}$   & Transmitting antenna gain & $30$~dBi, $30$~dBi, $30$~dBi  \\ \hline
$\nu_{\mathrm{LU}}$, $\nu_{\mathrm{FU}}$, $\nu_{\mathrm{LF}}$   & Wavelength & $0.015$~m, $0.015$~m, $0.015$~m  \\ \hline
$\zeta_U$, $\zeta_F$   & Rain attenuation   & $-2$~dB, $0$~dB     \\ \hline
$\sigma_{\mathrm{U}}^2$, $\sigma_{\mathrm{F}}^2$     & Noise power   & $-94$~dBm, $-84$~dBm  \\  \hline
% $B_{\mathrm{LU}}$, $B_{\mathrm{FU}}$, $B_{\mathrm{LF}}$     & Bandwidth   &$100$~MHz, $100$~MHz, $1$~GHz,  \\  \hline
$ (\Omega,b_0,m)$ & Parameters of SR fading   & $ (1.29,0.158,19.4)$ \\ \hline
\end{tabular}
}
\end{table*}

\subsection{Spatial Distribution Model}
We consider a LEO satellite constellation with $N_L$ satellites deployed on a spherical surface with radius $R_{\mathrm{sat}} = R_{\oplus} + h_{\mathrm{sat}}$, where $R_{\oplus}$ is the Earth's radius and $h_{\mathrm{sat}}$ is the altitude of satellites. The positions of the leader satellites follow a homogeneous BPP on the sphere. Note that we study the downlink communication performance of the satellite network from the perspective of a typical user, while satellites simultaneously serve multiple users. Under the homogeneous BPP modeling assumption, the performance of the typical user is representative of the average performance across a large population of users.

\par
The nearest leader satellite is associated with the ground user to provide communication services. We also refer to the nearest leader satellite to the user as the associated leader satellite. In addition, we consider that some leader satellites are equipped with follower satellites. Assume that each leader satellite is surrounded by $N_{\mathrm{F}}$ follower satellites. These follower satellites are deployed in a spherical cap centered on the leader satellite. The positions of the $N_{\mathrm{F}}$ follower satellites assigned to each leader satellite independently follow a homogeneous BPP within the spherical cap. It is worth noting that we consider a general scenario where not all leader satellites are equipped with followers. This allows for a comparison of the communication performance between a satellite cluster with followers and an individual leader satellite. 

\par
Furthermore, we denote the maximum communication distance between a leader satellite and a user as $d_{\mathrm{max}}$. Given that satellites are distributed on a sphere or spherical cap, using the central angle is more convenient for expression.
\begin{definition}[Central Angle] \label{definition1}
    The central angle is defined as the angle formed by the lines connecting two devices to the center of the Earth \cite{wang2022ultra}.
\end{definition}
Thus, the maximum central angle between a leader satellite and a user can be expressed as
\begin{equation}
    \theta_{\max} = \arccos \left( \frac{R_{\oplus}^2 + R_{\mathrm{sat}}^2 - \left(d_{\mathrm{max}} \right)^2 }{2R_{\oplus} R_{\mathrm{sat}}} \right).
\end{equation}
To ensure that all followers around the leader have a line-of-sight connection to the user, the following inequality needs to be satisfied:
\begin{equation}
    \theta_{\mathrm{max}} \leq \arccos\left( \frac{R_{\oplus}}{R_{\mathrm{sat}}} \right) - \theta_{\mathrm{cap}},
\end{equation}
where $\theta_{\mathrm{cap}}$ is the central angle of the spherical cap where the follower satellites are located.

\subsection{Channel Model}
This subsection models the space-ground links from the leader satellite to the user and from the follower satellite to the user, as well as the inter-satellite link from the leader to the follower. Let the received signal-to-noise ratio (SNR) of these three types of links be denoted as $\mathrm{SNR}_{\mathrm{LU}}$, $\mathrm{SNR}_{\mathrm{FU}}$, and $\mathrm{SNR}_{\mathrm{LF}}$, respectively. We consider that the channel links are experienced with large-scale path loss, rain attenuation, and small-scale fading. Therefore, the received SNR can be expressed as \cite{talgat2020stochastic}
\begin{equation} \label{SNR}
\mathrm{SNR}_{Q_1Q_2} = \rho_{Q_1Q_2} G_{Q_1Q_2} \zeta_{Q_2} W_{Q_2} \left( \frac{\nu_{Q_1Q_2}}{4\pi r_{Q_1Q_2} \sigma_{Q_2}} \right)^2,
\end{equation}
where $Q_1Q_2 \in \{ {\mathrm{LU}}, {\mathrm{FU}}, {\mathrm{LF}} \}$. Note that we distinguish between the types of transmitters and receivers using subscripts $Q_1$ and $Q_2$. For example, when the transmitter is the leader satellite and the receiver is the user, $Q_1Q_2$ is replaced by ${\mathrm{LU}}$. In (\ref{SNR}), $\rho_{Q_1Q_2}$, $G_{Q_1Q_2}$, $\nu_{Q_1Q_2}$, $r_{Q_1Q_2}$ denote the transmit power, antenna gain, wavelength, and distance between the transmitter and receiver, respectively. $\sigma_{Q_2}^2$ represents the thermal noise at the receiver. $\zeta_U$ denotes the rain attenuation factor for space-to-ground links, and we assume $\zeta_F = 1$ due to the minimal impact of rain attenuation on inter-satellite links.

To simplify the expression of the received SNR, we define the composite parameter
\begin{equation}
\xi_{Q_1Q_2} = \rho_{Q_1Q_2} G_{Q_1Q_2} \zeta_{Q_2}
\left( \frac{\nu_{Q_1Q_2}}{4\pi \sigma_{Q_2}} \right)^2,
\end{equation} 
By substituting this definition into \eqref{SNR}, the received SNR can be rewritten as
\begin{equation}
\mathrm{SNR}_{Q_1Q_2} = \xi_{Q_1Q_2} \, 
 \frac{W_{Q_2}}{r_{Q_1Q_2}^2}.
\end{equation}

\par
Next, the power of small-scale fading of the space-to-ground link is denoted as $W_{\mathrm{U}}$. The small-scale fading is modeled by the shadowed Rician (SR) distribution, which is considered one of the most accurate models for aerial/space-to-ground channel links \cite{huang2021uplink}. The cumulative distribution function (CDF) of $W_{\mathrm{U}}$ is given by \cite{jung2022performance}:
\begin{equation}\label{F_WU}
\begin{split}
    F_{W_{\mathrm{U}}} (w) & = \left( \frac{2b_0 m}{2b_0 m + \Omega} \right)^m \sum_{z=0}^{\infty} \frac{(m)_z}{z! \, \Gamma(z+1)} \\
    & \times \left( \frac{\Omega}{2b_0 m + \Omega} \right)^z \Gamma_l \left( z+1 , \frac{w}{2b_0} \right),
\end{split}
\end{equation}
where $\Omega$, $2b_0$, and $m$ denote the average power of the LoS component, the average power of the multi-path component excluding the LoS component, and the Nakagami parameter, respectively. $\Gamma(z+1)$ represents the Gamma function, and $(m)_z$ is the Pochhammer symbol. $\Gamma_l \left( \cdot , \cdot \right)$ is the lower incomplete Gamma function:
\begin{equation}
    \Gamma_l \left( z+1 , \frac{w}{2b_0} \right) = \int_0^{\frac{w}{2b_0}} t^{z} \, e^{-t} \mathrm{d} t.
\end{equation}

\par
Since processing the probability density function (PDF) of $W_{\mathrm{U}}$ is challenging, authors in \cite{jia2021uplink} approximated it as
\begin{equation}
    f_{W_{\mathrm{U}}} (w) \approx \frac{1}{m_2^{m_1} \Gamma \left( m_1 \right)} w^{m_1 - 1} e^{-\frac{w}{m_2} },
\end{equation}
where $\Gamma (\cdot)$ denotes the Gamma function. $m_1$ represents the shape parameter
\begin{equation}
    m_1 = \frac{m \left( 2b_0 + \Omega \right)^2}{4 m b_0^2 + 4 m b_0 \Omega + \Omega^2},
\end{equation}
and $m_2$ represents the scale parameter
\begin{equation}
    m_2 = \frac{4 m b_0^2 + 4 m b_0 \Omega + \Omega^2}{m \left( 2b_0 + \Omega \right)}.
\end{equation}

\par
Considering that the shadowing effect, multipath effect, and pointing errors in the leader-follower link are not significant, the small-scale fading of the inter-satellite link is neglected. As stated, frequency coordination within the satellite cluster is not difficult. Additionally, the wide bandwidth of space-ground communication allows for effective frequency allocation. Therefore, the impact of interference on performance will be discussed in future work.

\section{Performance Analysis}
The performance of the proposed LEO satellite system is characterized in terms of outage probability and average data rate. To facilitate the analysis, the analytical expressions of contact angle distributions are first derived.

\subsection{Contact Angle Distribution}
This section analyzes the distribution of communication distances from satellites that provide services to users. Considering the spherical modeling, it is more convenient to use the central angle rather than the Euclidean distance. Based on the definition of the central angle in Definition~\ref{definition1}, we provide the following definition.

\begin{definition}[Contact Angle] \label{definition2}
The central angle between the user and the associated leader satellite is referred to as the contact angle \cite{al2021analytic}. 
\par
When the associated leader satellite is equipped with follower satellites, the central angle between any follower satellite within the cluster and the user is also referred to as the contact angle.
\end{definition}

We first derive the distribution of contact angle from the simple case without follower satellites. 

\begin{lemma} \label{lemma1}
Denote the contact angle between the associated leader satellite and the user as $\theta_{\mathrm{LU}}$. The PDF of the distribution of $\theta_{\mathrm{LU}}$ is given by:
\begin{equation}
f_{\theta_{\mathrm{LU}}}(\theta) = \frac{N_{\mathrm{L}} \sin \theta}{2} \left( \frac{1 + \cos \theta}{2} \right)^{N_{\mathrm{L}} - 1}, \ 0 \leq \theta_{\mathrm{LU}} \leq \theta_{\max}.
\end{equation}
\end{lemma}
\begin{IEEEproof}
    See Appendix~\ref{app:lemma1}.
\end{IEEEproof}

Next, assuming that the associated leader satellite is equipped with followers, the contact angle between any follower and the user is denoted as $\theta_{\mathrm{FU}}$. Unlike Lemma~\ref{lemma1}, the derivation of the distribution of $\theta_{\mathrm{FU}}$ requires a case-by-case analysis depending on whether the user is located beneath the spherical cap where the followers are distributed. In other words, the distribution of $\theta_{\mathrm{FU}}$  depends on the value of $\theta_{\mathrm{LU}}$.

\begin{lemma}\label{lemma2} 
\textbf{Case 1}: When $\theta_{\mathrm{LU}} \in [\theta_{\mathrm{cap}}, \theta_{\max}]$, the user is located outside the spherical cap. In this case, the PDF of the distribution of $\theta_{\mathrm{FU}}$ is given by
\begin{equation}
\begin{split}
& f_{\theta_{\mathrm{FU}}}(\theta) = \frac{ \cos \theta_{\mathrm{cap}} \sec^2(\theta_{\mathrm{LU}} - \theta) }{\pi (1 - \cos\theta_{\mathrm{cap}})} \\
&\times \arcsin \left( \sqrt{ \sin^2 \theta_{\mathrm{cap}} - \cos^2 \theta_{\mathrm{cap}} \tan^2(\theta_{\mathrm{LU}} - \theta) } \right).
\end{split}
\end{equation}
where $\theta \in [\theta_{\mathrm{cap}} - \theta_{\mathrm{LU}}, \theta_{\mathrm{cap}} + \theta_{\mathrm{LU}}]$.

\par
\textbf{Case 2}: When $\theta_{\mathrm{LU}} \in [0, \theta_{\mathrm{cap}})$, the user is located inside the spherical cap. In this case:
\begin{itemize}
\item For $\theta \in [\theta_{\mathrm{cap}} - \theta_{\mathrm{LU}}, \theta_{\mathrm{cap}} + \theta_{\mathrm{LU}}]$, the PDF of  is the same as that in Case 1.
\item For $\theta \in [0, \theta_{\mathrm{cap}} - \theta_{\mathrm{LU}}]$, the PDF of the distribution of $\theta_{\mathrm{FU}}$ is given by
\end{itemize}
\begin{equation}
\begin{split}
&f_{\theta_{\mathrm{FU}}}(\theta) = \frac{ \cos \theta_{\mathrm{cap}} }{\pi (1 - \cos\theta_{\mathrm{cap}})} \\
&\times \Bigg[ \frac{ \arcsin \left( \sqrt{ \sin^2 \theta_{\mathrm{cap}} - \cos^2 \theta_{\mathrm{cap}} \tan^2(\theta_{\mathrm{LU}} - \theta) } \right)}{\cos^2(\theta_{\mathrm{LU}} - \theta)} \\
&- \frac{ \arcsin \left( \sqrt{ \sin^2 \theta_{\mathrm{cap}} - \cos^2 \theta_{\mathrm{cap}} \tan^2(\theta_{\mathrm{LU}} + \theta) } \right)}{\cos^2(\theta_{\mathrm{LU}} + \theta)} \Bigg].
\end{split}
\end{equation}
\end{lemma}
\begin{IEEEproof}
    See Appendix~\ref{app:lemma2}.
\end{IEEEproof}

As shown in Lemma~\ref{lemma2}, the explicit expression of the PDF can only be derived when the ranges of $\theta_{\mathrm{LU}}$ and $\theta$ are known, and the resulting expression is relatively complex. These reason limits the expressions' applicability in subsequent use as a lemma. Therefore, we consider the distribution of contact angles between the user and the nearest and farthest positions within the spherical cap region where the follower satellites are distributed, which serve as the lower and upper bounds of $\theta_{\mathrm{FU}}$, respectively.

\begin{lemma}\label{lemma3}
Denote the minimum and maximum contact angles between the user and follower satellites as $\theta_{\mathrm{FU}}^{\mathrm{min}}$ and $\theta_{\mathrm{FU}}^{\mathrm{max}}$. The PDF of the distribution of $\theta_{\mathrm{FU}}^{\mathrm{min}}$ is given as:
\begin{equation}
\begin{split}
& f_{\theta_{\mathrm{FU}}^{\mathrm{min}}}(\theta) = \\
&  
\begin{cases}
\left(1 - \left( \dfrac{1 + \cos(\theta_{\mathrm{cap}})}{2} \right)^{N_{\mathrm{L}}}\right)\cdot \delta(\theta), \ \mathrm{when} \ \theta = 0, \\
\dfrac{N_{\mathrm{L}} \sin(\theta + \theta_{\mathrm{cap}})}{2} \left( \dfrac{1 + \cos(\theta + \theta_{\mathrm{cap}})}{2} \right)^{N_{\mathrm{L}} - 1}, \\
\ \ \ \ \ \ \ \ \ \ \ \ \ \ \ \ \ \ \ \ \ \ \ \ \ \
\mathrm{when} \ \theta \in (0,\theta_{\max} - \theta_{\mathrm{cap}}], 
\end{cases}
\end{split}
\end{equation}
where $\delta(\cdot)$ denotes the Dirac delta function, capturing the discontinuity at $\theta = 0$. The PDF of the distribution of $\theta_{\mathrm{FU}}^{\mathrm{max}}$ is given as:
\begin{equation}
    f_{\theta_{\mathrm{FU}}^{\mathrm{max}}}(\theta) =
    \dfrac{N_{\mathrm{L}} \sin(\theta - \theta_{\mathrm{cap}})}{2} \left( \dfrac{1 + \cos(\theta - \theta_{\mathrm{cap}})}{2} \right)^{N_{\mathrm{L}} - 1},
\end{equation}
when $\theta \in \left[\theta_{\mathrm{cap}},\ \theta_{\max} + \theta_{\mathrm{cap}} \right]$.
% \begin{equation}
%     \begin{split}
%     f_{\theta_{\mathrm{FU}}^{\mathrm{max}}}(\theta)=
%     \begin{cases}
%     \left(1 - \left( \frac{1 + \cos(\theta_{\mathrm{cap}})}{2} \right)^{N_{\mathrm{L}}} \right)\cdot \delta(\theta), \theta = 0 \\
%     \frac{N_{\mathrm{L}} \sin(\theta - \theta_{\mathrm{cap}})}{2} \left( \frac{1 + \cos(\theta - \theta_{\mathrm{cap}})}{2} \right)^{N_{\mathrm{L}} - 1},\theta\in [2\theta_{\mathrm{cap}},\theta_{\mathrm{max}})\\
%     \left(1-F_{\theta_{\mathrm{LU}}}(\theta_{\mathrm{max}} - \theta_{\mathrm{cap}})\right)\cdot \delta\left(\theta\right),\theta=\theta_{\mathrm{max}}.
%     \end{cases}
%     \end{split}
% \end{equation}
\end{lemma}
\begin{IEEEproof}
    See Appendix~\ref{app:lemma3}.
\end{IEEEproof}

In Lemma~\ref{lemma3}, the expressions of $\theta_{\mathrm{FU}}^{\mathrm{min}}$ and $\theta_{\mathrm{FU}}^{\mathrm{max}}$ in Lemma~\ref{lemma3} can be explicitly derived without requiring knowledge of any variable's range. In addition, they do not involve any integrals, and thus can serve as low-complexity alternatives to the expressions in Lemma~\ref{lemma2}. As $\theta_{\mathrm{cap}}$ decreases, the upper and lower bounds gradually converge, making the approximation using the expressions in Lemma~\ref{lemma3} increasingly accurate.

\subsection{Outage Probability}
This subsection shows the results of the outage probability, whose definition is given as follows.

\begin{definition}[Outage Probability] \label{definition3}
When no follower satellite provides coverage, the outage probability is defined as the probability that the SNR received by the user from the associated leader satellite is below a predefined threshold $\gamma_{\mathrm{th}}$.
\par
For a satellite cluster equipped with followers, the outage probability is defined as the probability that the SNR received by the user from all satellites in the cluster is lower than $\gamma_{\mathrm{th}}$.
\end{definition}

According to the definition, a satellite cluster is more likely to establish a stable link with the user compared to a single leader satellite. Similarly, we start with the derivation of outage probability from the simple case without follower satellites.

\begin{theorem} \label{theorem1}
The outage probability of the associated leader satellite to the user link is given as 
% \begin{equation}
% \begin{split}
%     P_{\mathrm{out}}^{\mathrm{LU}}= \int_{0}^{\theta_{\max}} 
%     F_{W_{\mathrm{U}}} \left( \frac{\gamma_{\mathrm{th}}}{\rho_{\mathrm{LU}}G_{\mathrm{LU}}\zeta_{\mathrm{U}}} \left(\frac{4\pi\sigma_{\mathrm{U}}r_{\mathrm{LU}}(\theta)}{\nu_{\mathrm{LU}}}\right)^2 \right) & \\
%     \times     f_{\theta_{\mathrm{LU}}}(\theta) \, \mathrm{d} \theta, & 
% \end{split}
% \end{equation}

\begin{equation}
\begin{split}
    P_{\mathrm{out}}^{\mathrm{LU}}= \int_{0}^{\theta_{\max}} 
    F_{W_{\mathrm{U}}} \left( \frac{\gamma_{\mathrm{th}}r_{\mathrm{LU}}^2(\theta)}{\xi_{\mathrm{LU}}}  \right) 
   f_{\theta_{\mathrm{LU}}}(\theta) \, \mathrm{d} \theta, 
\end{split}
\end{equation}

where $\gamma_{\mathrm{th}}$ is the predefined coverage threshold. $r_{\mathrm{LU}}(\theta)$ denotes the Euclidean distance between the associated leader satellite and the typical user,
\begin{equation}\label{r_LU}
    r_{\mathrm{LU}}(\theta) = \sqrt{R_{\mathrm{sat}}^2 + R_{\oplus}^2 - 2 R_{\mathrm{sat}} R_{\oplus} \cos\theta}.
\end{equation}
\end{theorem}
\begin{IEEEproof}
See Appendix~\ref{app:theorem1}.
\end{IEEEproof}

\par
Due to spatial diversity, the small-scale fadings from different follower satellites and the leader satellite are independently distributed. Therefore, the outage probability from the satellite cluster to the user can be derived accordingly.

\begin{theorem}\label{theorem2}
% The outage probability of a satellite cluster including $N_{\mathrm{F}}$ follower satellites is given as:
% The outage probability of a satellite cluster including $N_F$ follower satellites is given as:
% \begin{equation}
%     P_{\mathrm{out}}^{\mathrm{Cluster}} =\int_{0}^{\theta_{\mathrm{max}}} \left(P_{\mathrm{out},i}^{\mathrm{FU}}\left(\theta\right) \right)^{N_{\mathrm{F}}}\cdot f_{\theta_{\mathrm{LU}}}(\theta) \ \mathrm{d} \theta,
% \end{equation}
% where $P_{\mathrm{out},i}^{\mathrm{FU}}\left(\theta \right)$ denotes the outage probability between the user and the $i$th follower satellite,
% \begin{equation}
% \begin{split}
%     P_{\mathrm{out},i}^{\mathrm{FU}} \left(\theta\right)
%     & = \int_{\theta_1}^{\theta_2} 
%     F_{W_{\mathrm{U}}} \left( \frac{\gamma_{\mathrm{th}}}{\rho_{\mathrm{FU}}G_{\mathrm{FU}}\zeta_{\mathrm{U}}} \cdot \left(\frac{4\pi\sigma_{\mathrm{U}}r_{\mathrm{FU}} \left(\theta\right)}{\nu_{\mathrm{FU}}}\right)^2 \right) f_{\theta_{\mathrm{FU}}}\left(\theta\right)\,  \mathrm{d}\theta.
%     \end{split}
% \end{equation}
% Here, $r_{\mathrm{FU}}$ represents the distance between the follower satellite and the user,
% \begin{equation}
%     r_{\mathrm{FU}}\left(\theta\right) = \sqrt{R_{\oplus}^2 + R_{\mathrm{sat}}^2 - 2 R_{\oplus} R_{\mathrm{sat}} \cos(\theta)}
% \end{equation}

Given that a satellite cluster equipped with $N_{\mathrm{F}}$ follower satellites, the outage probability of the cluster is
\begin{equation}
P_{\mathrm{out}}^{\mathrm{Cluster}} =\int_{0}^{\theta_{\mathrm{max}}} \left(P_{\mathrm{out},i}^{\mathrm{FU}}\left(\theta\right)\right)^{N_{\mathrm{F}}}P_{\mathrm{Cond}}^{\mathrm{LU}}\left(\theta\right) f_{\theta_{\mathrm{LU}}}(\theta) \ \mathrm{d} \theta.
\end{equation}
$P_{\mathrm{Cond}}^{\mathrm{LU}}\left(\theta \right)$ is the conditional outage probability from the leader satellite given $\theta_{\mathrm{LU}}$, and it can be expressed as
% \begin{equation}
% \begin{split}
% P_{\mathrm{Cond}}^{\mathrm{LU}}\left(\theta \right) = F_{W_{\mathrm{U}}} \left( \frac{\gamma_{\mathrm{th}}}{\rho_{\mathrm{LU}}G_{\mathrm{LU}}\zeta_{\mathrm{U}}} \left(\frac{4\pi\sigma_{\mathrm{U}}r_{\mathrm{LU}}(\theta)}{\nu_{\mathrm{LU}}}\right)^2 \right).
% \end{split}
% \end{equation}

\begin{equation}
\begin{split}
P_{\mathrm{Cond}}^{\mathrm{LU}}\left(\theta \right) = F_{W_{\mathrm{U}}} \left( \frac{\gamma_{\mathrm{th}}r_{\mathrm{LU}}^2(\theta)}{\xi_{\mathrm{LU}}}  \right).
\end{split} 
\end{equation}
$P_{\mathrm{out},i}^{\mathrm{FU}}\left(\theta \right)$ denotes the outage probability for the $i$-th follower satellite to the user link, which can be expressed as,
% \begin{equation}\label{P_out_fu_i}
% \begin{split}
% & P_{\mathrm{out},i}^{\mathrm{FU}} \left(\theta\right) = \int_{0}^{2\pi} \int_{0}^{\theta_{\mathrm{cap}}} \frac{\sin \psi}{2\pi\left(1 - \cos\theta_{\mathrm{cap}}\right)}
% \\
% & \times 
%     F_{W_{\mathrm{U}}} \left( \frac{\gamma_{\mathrm{th}}}{\rho_{\mathrm{FU}}G_{\mathrm{FU}}\zeta_{\mathrm{U}}} \left(\frac{4\pi\sigma_{\mathrm{U}}r_{\mathrm{FU}} \left(\theta, \psi, \varphi \right)}{\nu_{\mathrm{FU}}}\right)^2 \right) \, \mathrm{d}\psi \mathrm{d}\varphi.
%     \end{split}
% \end{equation}

\begin{equation}\label{P_out_fu_i}
\begin{split}
P_{\mathrm{out},i}^{\mathrm{FU}} \left(\theta\right) =& \int_{0}^{2\pi} \int_{0}^{\theta_{\mathrm{cap}}} \frac{\sin \psi}{2\pi\left(1 - \cos\theta_{\mathrm{cap}}\right)}
\\
& \times 
    F_{W_{\mathrm{U}}} \left( \frac{\gamma_{\mathrm{th}}r_{\mathrm{FU}}^2\left(\theta, \psi, \varphi \right)}{\xi_{\mathrm{FU}}}  \right) \, \mathrm{d}\psi \mathrm{d}\varphi.
    \end{split}
\end{equation}
where $r_{\mathrm{FU}}$ represents the Euclidean distance between the follower satellite and the user,
\begin{equation}\label{r_FU}
\begin{aligned}
r_{\mathrm{FU}}\left(\theta, \psi, \varphi\right) & = \Big( R_{\oplus}^2 + R_{\mathrm{sat}}^2 
- 2 R_{\oplus} R_{\mathrm{sat}} \\
& \times \big( 
\sin\theta \sin\psi \cos\varphi + \cos\theta \cos\psi 
\big) \Big)^{1/2}.
\end{aligned}
\end{equation}

% \begin{equation}
%     P_{\mathrm{out}}^{\mathrm{Cluster}} =\int_{0}^{\theta_{\mathrm{max}}} \left(P_{\mathrm{out},i}^{\mathrm{FU}}\left(\theta\right) \right)^{N_{\mathrm{F}}}\cdot f_{\theta_{\mathrm{LU}}}(\theta) \ \mathrm{d} \theta,
% \end{equation}
% where $P_{\mathrm{out},i}^{\mathrm{FU}}\left(\theta \right)$ denotes the outage probability between the user and the $i$th follower satellite,
% \begin{equation}\label{P_out_fu_i}
% \begin{split}
% & P_{\mathrm{out},i}^{\mathrm{FU}} \left(\theta\right) = \int_{0}^{2\pi} \int_{0}^{\theta_{\mathrm{max}}} \frac{\sin \psi}{2\pi\left(1 - \cos\theta_{\mathrm{cap}}\right)}
% \\
% & \times 
%     F_{W_{\mathrm{U}}} \left( \frac{\gamma_{\mathrm{th}}}{\rho_{\mathrm{F}}G_{\mathrm{FU}}\zeta_{\mathrm{U}}} \left(\frac{4\pi\sigma_{\mathrm{U}}r_{\mathrm{FU}} \left(\theta, \psi, \varphi \right)}{\nu_{\mathrm{FU}}}\right)^2 \right) \, \mathrm{d}\psi \mathrm{d}\varphi.
%     \end{split}
% \end{equation}
% Here, $r_{\mathrm{FU}}$ represents the distance between the follower satellite and the user,
% \begin{equation}
% \begin{aligned}
% r_{\mathrm{FU}}\left(\theta, \psi, \varphi\right) & = \Big( R_{\oplus}^2 + R_{\mathrm{sat}}^2 
% - 2 R_{\oplus} R_{\mathrm{sat}} \\
% & \times \big( 
% \sin\theta \sin\psi \cos\varphi + \cos\theta \cos\psi 
% \big) \Big)^{1/2}.
% \end{aligned}
% \end{equation}
\end{theorem}
\begin{IEEEproof}
See Appendix~\ref{app:theorem2}.
\end{IEEEproof}

Since the expression in Theorem~\ref{theorem1} involves a triple integral, we provide upper and lower bounds of the outage probability based on Lemma~\ref{lemma3}, which involves only a single integral.

\begin{corollary} \label{corollary1}
A lower bound of the outage probability of the cluster is given as
\begin{equation} \label{outagelower}
\begin{split}
& P_{\mathrm{out}}^{ \mathrm{lower}} = \left(1 - \left( \dfrac{1 + \cos(\theta_{\mathrm{cap}})}{2} \right)^{N_{\mathrm{L}}}\right)  \\
& \times 
\left(\int _{0}^{\theta_{\mathrm{cap}}}P_{\mathrm{Cond}}^{\mathrm{LU}}\left(\theta \right)\mathrm{d} \theta\right) P_{\mathrm{FU}}^{\mathrm{bound}}(0) \\
& + \int_{0^+}^{\theta_{\mathrm{max}}-\theta_{\mathrm{cap}}} P_{\mathrm{FU}}^{\mathrm{bound}}(\theta)P_{\mathrm{Cond}}^{\mathrm{LU}}\left(\theta+\theta_{\mathrm{cap}} \right) f_{\theta_{\mathrm{FU}}^{\mathrm{min}}}(\theta) \mathrm{d}\theta.
\end{split}
\end{equation}
An upper bound of the outage probability of the cluster is given as
\begin{equation} \label{outageupper}
P_{\mathrm{out}}^{\mathrm{upper}} =  \int_{\theta_{\mathrm{cap}}}^{\theta_{\mathrm{max}}+\theta_{\mathrm{cap}}} \! \! \! \! \! \! \! \! \! \! \! \! \! \! \! P_{\mathrm{FU}}^{\mathrm{bound}}(\theta)P_{\mathrm{Cond}}^{\mathrm{LU}}\left(\theta-\theta_{\mathrm{cap}} \right) f_{\theta_{\mathrm{FU}}^{\mathrm{max}}}(\theta) \mathrm{d}\theta.
\end{equation}

\par
$P_{\mathrm{FU}}^{\mathrm{bound}}(\theta)$ in (\ref{outagelower}) and (\ref{outageupper}) are defined as
% \begin{equation}
% \begin{split}
% P_{\mathrm{FU}}^{\mathrm{bound}} (\theta)&=\Bigg(F_{W_{\mathrm{U}}}\Bigg( \frac{\gamma_{\mathrm{th}}}{\rho_{\mathrm{FU}}G_{\mathrm{FU}}\zeta_{\mathrm{U}}} \\
% &\times\left(\frac{4\pi\sigma_{\mathrm{U}}r_{\mathrm{FU}}^{\mathrm{bound}} \left(\theta\right)}{\nu_{\mathrm{FU}}}\right)^2 \Bigg)\Bigg)^{N_{\mathrm{F}}},
% \end{split}
% \end{equation}

\begin{equation}
\begin{split}
P_{\mathrm{FU}}^{\mathrm{bound}} (\theta)&=\Bigg(F_{W_{\mathrm{U}}}\left( \frac{\gamma_{\mathrm{th}}(r_{\mathrm{FU}}^{\mathrm{bound}}\left(\theta\right))^2}{\xi_{\mathrm{FU}}}  \right)\Bigg)^{N_{\mathrm{F}}},
\end{split}
\end{equation}
where $r_{\mathrm{FU}}^{\mathrm{bound}}\left(\theta\right)$ represents the Euclidean distance between the user and the closest or farthest follower satellites,
\begin{equation}\label{r_FU_bound}
\begin{split}
r_{\mathrm{FU}}^{\mathrm{bound}}\left(\theta\right) = \sqrt{ R_{\oplus}^2 + R_{\mathrm{sat}}^2 - 2 R_{\oplus} R_{\mathrm{sat}} \cos\theta}.
\end{split}
\end{equation}
\end{corollary}

\subsection{Average Data Rate}
This subsection shows the results of the average data rate, whose definition is given as follows.

\begin{definition}[Average Data Rate of Leader Satellite]\label{definition4}
According to the Shannon capacity formula, the average data rate from the leader satellite is mathematically defined as
\begin{equation}
    \mathcal{R}_{\mathrm{LU}} = B_{\mathrm{LU}} \log_2 \left( 1 + {\mathrm{SNR}}_{\mathrm{LU}} \right),
\end{equation}
where $B_{\mathrm{LU}}$ is the bandwidth of the leader satellite-user link.
\end{definition}

Next, we start with the derivation of the average data rate from the simple case without follower satellites.

\begin{theorem}\label{theorem3}
The average data rate of the associated leader satellite is given as:
% \begin{equation}\label{R_LU}
% \begin{split}
% & \mathcal{R}_{\mathrm{LU}} = \int_{0}^{\theta_{\mathrm{max}}} \int_0^{\infty} B_{\mathrm{LU}} f_{W_{\mathrm{U}}}(w) f_{\theta_{\mathrm{LU}}}(\theta)  \\
% & \times \log_2\left(1 + \rho_{\mathrm{LU}} G_{\mathrm{LU}} \zeta_{\mathrm{U}} w \left( \frac{\nu_{\mathrm{LU}}}{4\pi \sigma_{\mathrm{U}}r_{\mathrm{LU}}(\theta)} \right)^2 \right) \mathrm{d}w \mathrm{d}\theta.
% \end{split}
% \end{equation}

\begin{equation}\label{R_LU}
\begin{split}
\mathcal{R}_{\mathrm{LU}} & =  \int_{0}^{\theta_{\mathrm{max}}} \int_0^{\infty} B_{\mathrm{LU}} f_{W_{\mathrm{U}}}(w)  \\
& \times f_{\theta_{\mathrm{LU}}}(\theta) \log_2\left(1 +  
 \frac{\xi_{\mathrm{LU}}w}{r_{\mathrm{LU}}^2(\theta)}\right) \mathrm{d}w \mathrm{d}\theta.
\end{split}
\end{equation}
\end{theorem}
\begin{IEEEproof}
See Appendix~\ref{app:theorem3}.
\end{IEEEproof}

\par
Next, after the follower satellites are deployed, the leader satellite allocates most of its power to transmitting signals to the users, while the remaining power is used to establish links with the follower satellites. The follower satellites act as relays, forwarding the leader's messages to the users.

\begin{definition}[Average Data Rate of Cluster] \label{definition5}
The average data rate of a satellite cluster $\mathcal{R}_{\mathrm{LFs}}$ is defined as the summation of the average data rate from all satellite-user links. Mathematically, it is expressed as
\begin{equation}\label{def:R_c}
\mathcal{R}_{\mathrm{LFs}} = \sum_{i=1}^{N_{\mathrm{F}}} \min \left\{  \mathcal{R}_{\mathrm{LF},i},  \mathcal{R}_{\mathrm{FU},i} \right\}+\mathcal{R}_{\mathrm{LU}},
\end{equation}
where $\mathcal{R}_{\mathrm{LF},i}$ and $\mathcal{R}_{\mathrm{FU},i}$ denote the data rates of the leader-to-$i$-th-follower and $i$-th-follower-to-user links, respectively.
\end{definition}

\par
Since the distance from the leader to follower satellites is much shorter than the distance from the follower satellite to the user, only a small portion of the power needs to be allocated for the inter-satellite links. Although the data is relayed, the total power of all satellites in the cluster is significantly higher than that of a single leader satellite, making it possible to improve the average data rate with the assistance of the follower satellites. Next, the average data rate of the satellite cluster is derived in the following theorem.

\begin{theorem}\label{theorem4}
The average data rate of the satellite cluster is given as
\begin{equation}
\begin{split}
\mathcal{R}_{\mathrm{LFs}}=N_{\mathrm{F}} \int_{0}^{\infty}E(w)f_{W_{\mathrm{U}}}\left(w\right) \mathrm{d}w+\mathcal{R}_{\mathrm{LU}},
\end{split}
\end{equation}
where $\mathcal{R}_{\mathrm{LU}}$ is given in (\ref{R_LU}) and 
\begin{equation} \label{E(w)}
\begin{split}
& E\left(w\right)=\int_{0}^{2\pi} \int_{0}^{\theta_{\mathrm{cap}}} \int_{0}^{\theta_{\mathrm{max}}}\frac{\sin \psi \, f_{\theta_{\mathrm{LU}}}(\theta)}{2\pi\left(1 - \cos\theta_{\mathrm{cap}}\right)}  \\
&\times \min\left\{ \mathcal{R}_{\mathrm{LF},i}(\psi), \mathcal{R}_{\mathrm{FU},i}(w,\theta,\psi,\varphi) \right\} \mathrm{d}\theta\mathrm{d}\psi \mathrm{d}\varphi.
   \end{split}
\end{equation}
The data rates associated with the leader-to-follower and follower-to-user links for the $i$-th follower satellite are respectively expressed as
% \begin{sequation}
% \mathcal{R}_{\mathrm{LF},i}\left(\psi\right) = B_{\mathrm{LF}} \log_2\left(1 + \frac{\rho_{\mathrm{LF}} G_{\mathrm{LF}} \zeta_{\mathrm{F}} W_{\mathrm{F}}\left(\nu_{\mathrm{LF}}\right)^2}{\left(4\pi \sigma_{\mathrm{F}}r_{\mathrm{LF}}\left(\psi\right)\right)^2}\right),
% \end{sequation}

\begin{sequation}
\mathcal{R}_{\mathrm{LF},i}\left(\psi\right) = B_{\mathrm{LF}} \log_2\left(1 + \xi_{\mathrm{LF}} \, 
 \frac{W_{\mathrm{F}}}{r_{\mathrm{LF}}^2\left(\psi\right)}\right),
\end{sequation}

% \begin{sequation} \mathcal{R}_{\mathrm{FU},i}\left(w,\theta,\varphi,\psi\right) = B_{\mathrm{FU}} \log_2\left(1 + \frac{\rho_{\mathrm{FU}} G_{\mathrm{FU}} \zeta_{\mathrm{U}} w\left(\nu_{\mathrm{FU}}\right)^2}{\left(4\pi \sigma_{\mathrm{U}}r_{\mathrm{FU}}\left(\theta,\varphi,\psi\right)\right)^2} \right),
% \end{sequation}
\begin{sequation} \mathcal{R}_{\mathrm{FU},i}\left(w,\theta,\varphi,\psi\right) = B_{\mathrm{FU}} \log_2\left(1 + \xi_{\mathrm{FU}} \, 
 \frac{w}{r_{\mathrm{FU}}^2\left(\theta,\varphi,\psi\right)} \right),
\end{sequation}
where $r_{\mathrm{FU}}\left(\theta,\varphi,\psi\right)$ is defined in (\ref{r_FU}) and $r_{\mathrm{LF}}\left(\psi\right)$ is expressed as
\begin{equation}
r_{\mathrm{LF}}\left(\psi\right) = R_{\mathrm{sat}} \sqrt{2 (1 - \cos \psi)}.
\end{equation}
\end{theorem}
\begin{IEEEproof}
See Appendix~\ref{app:theorem4}.
\end{IEEEproof}

\par
Since the expression of $\mathcal{R}_{\mathrm{LFs}}$ involves a four-fold integral with high computational complexity, it is necessary to provide a low-complexity approximation. The following corollary provides upper and lower bounds for $\mathcal{R}_{\mathrm{LFs}}$.

\begin{corollary} \label{corollary2}
When $\theta_{\mathrm{cap}} \ll \theta_{\max}$, an upper bound of the average data rate of the satellite cluster is given as
% \begin{equation}
% \begin{split}
% &\mathcal{R}_{\mathrm{LFs}}^{\mathrm{upper}}
% = N_{\mathrm{F}} B_{\mathrm{FU}} \int_{0}^{\theta_{\mathrm{max}}-\theta_{\mathrm{cap}}} \int_0^{\infty}
%  f_{W_{\mathrm{U}}}(w) f_{\theta_{\mathrm{FU}}^{\mathrm{min}}}(\theta) \\
% & \times \log_2\left(1 + \frac{\rho_{\mathrm{FU}} G_{\mathrm{FU}} \zeta_{\mathrm{U}} w\left(\nu_{\mathrm{FU}}\right)^2}{\left(4\pi \sigma_{\mathrm{U}}r_{\mathrm{FU}}^{\mathrm{bound}}\left(\theta\right)\right)^2}\right)
% \mathrm{d}w\, \mathrm{d}\theta+\mathcal{R}_{\mathrm{LU}},
% \end{split}
% \end{equation}
\begin{equation} \label{rateupper}
\begin{split}
&\mathcal{R}_{\mathrm{LFs}}^{\mathrm{upper}}
= N_{\mathrm{F}} B_{\mathrm{FU}} \int_{0}^{\theta_{\mathrm{max}}-\theta_{\mathrm{cap}}} \int_0^{\infty}
 f_{W_{\mathrm{U}}}(w) f_{\theta_{\mathrm{FU}}^{\mathrm{min}}}(\theta) \\
& \times \log_2\left(1 + 
 \frac{\xi_{\mathrm{FU}}w}{(r_{\mathrm{FU}}^{\mathrm{bound}}\left(\theta\right))^2}\right)
\mathrm{d}w\, \mathrm{d}\theta+\mathcal{R}_{\mathrm{LU}},
\end{split}
\end{equation}
where $r_{\mathrm{FU}}^{\mathrm{bound}}\left(\theta\right)$ is defined in (\ref{r_FU_bound}). When $\theta_{\mathrm{cap}} \ll \theta_{\max}$, a lower bound of the average data rate of the satellite cluster is given as
% \begin{equation}
% \begin{split}
% &\mathcal{R}_{\mathrm{LFs}}^{\mathrm{lower}} = N_{\mathrm{F}} B_{\mathrm{FU}} \int_{\theta_{\mathrm{cap}}}^{\theta_{\mathrm{max}}+\theta_{\mathrm{cap}}} \int_0^{\infty} f_{W_{\mathrm{U}}}(w) f_{\theta_{\mathrm{FU}}^{\mathrm{max}}}(\theta) 
%  \\
% &\times \log_2\left(1 + \frac{\rho_{\mathrm{FU}} G_{\mathrm{FU}} \zeta_{\mathrm{U}} w\left(\nu_{\mathrm{FU}}\right)^2}{\left(4\pi \sigma_{\mathrm{U}}r_{\mathrm{FU}}^{\mathrm{bound}}\left(\theta\right)\right)^2}\right) \mathrm{d}w\, \mathrm{d}\theta+\mathcal{R}_{\mathrm{LU}}.
% \end{split}
% \end{equation}
\begin{equation} \label{ratelower}
\begin{split}
&\mathcal{R}_{\mathrm{LFs}}^{\mathrm{lower}} = N_{\mathrm{F}} B_{\mathrm{FU}} \int_{\theta_{\mathrm{cap}}}^{\theta_{\mathrm{max}}+\theta_{\mathrm{cap}}} \int_0^{\infty} f_{W_{\mathrm{U}}}(w) f_{\theta_{\mathrm{FU}}^{\mathrm{max}}}(\theta) 
 \\
&\times \log_2\left(1 + \frac{\xi_{\mathrm{FU}}w}{(r_{\mathrm{FU}}^{\mathrm{bound}}\left(\theta\right))^2}\right) \mathrm{d}w\, \mathrm{d}\theta+\mathcal{R}_{\mathrm{LU}}.
\end{split}
\end{equation}
\end{corollary}

\par
Since condition $\theta_{\mathrm{cap}} \ll \theta_{\max}$ is assumed, the data rate is mainly limited by the follower-to-user link. Combined with applying the result of Lemma~\ref{lemma3}, the derived result in Corollary~\ref{corollary2} is therefore an approximation.

\section{Numerical Results}
This section presents the numerical results of the outage probability and average data rate. Through Monte Carlo simulations, we verify the accuracy of the analytical results derived. Furthermore, we explore the impact of different follower deployment configurations on performance and design case studies to demonstrate the improvement in average data rate with the deployment of followers. Finally, unless otherwise specified, the parameter values used in numerical simulations are taken as the default values in Table~\ref{table1}.

\begin{figure}[ht]
\centering
\includegraphics[width=0.9\linewidth]{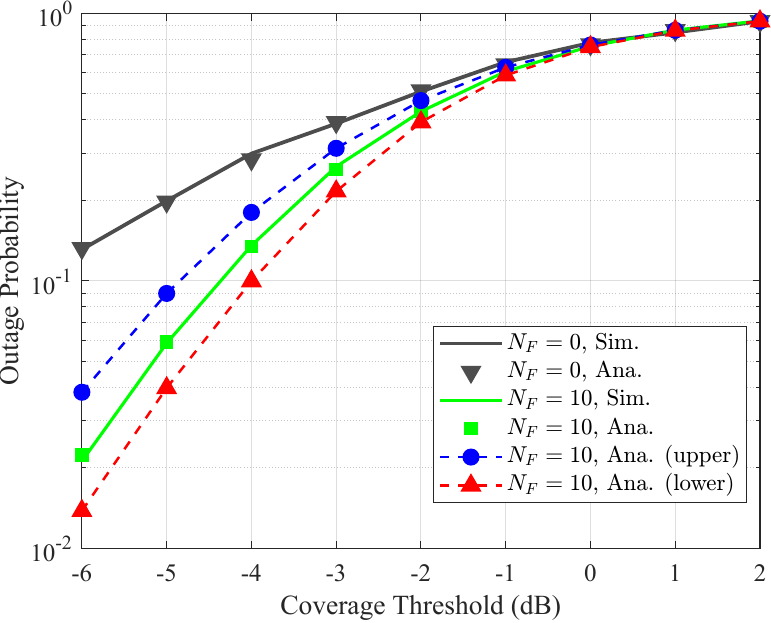}
\caption{Verification of the accuracy of outage probability.}
\label{fig:outage_vs_gamma_log}
\end{figure}

\subsection{Outage Probablity}
This subsection presents the numerical results of the outage probability, with Fig.~\ref{fig:outage_vs_gamma_log} primarily used to verify the accuracy of the derived results. In the labels of Fig.~\ref{fig:outage_vs_gamma_log}, lines represent Monte Carlo simulation results, while dots represent analytical results. The black triangular points are obtained from Theorem~\ref{theorem1}, the green squares from Theorem~\ref{theorem2}, and the other two types of dots are obtained from Corollary~\ref{corollary1}.
The overlap of the dots and lines verifies the accuracy of Theorem~\ref{theorem1} and Theorem~\ref{theorem2}. The results in Corollary~\ref{corollary1} offer a relatively tight upper and lower bound for the outage probability. As the coverage threshold increases, the gap between the upper and lower bounds and the simulation results gradually diminishes. Furthermore, Fig.~\ref{fig:outage_vs_hsat_log} also verifies the accuracy of Theorem~\ref{theorem1} and Theorem~\ref{theorem2} through the alignment of dots and lines.

\begin{figure}[ht]
\centering
\includegraphics[width=\linewidth]{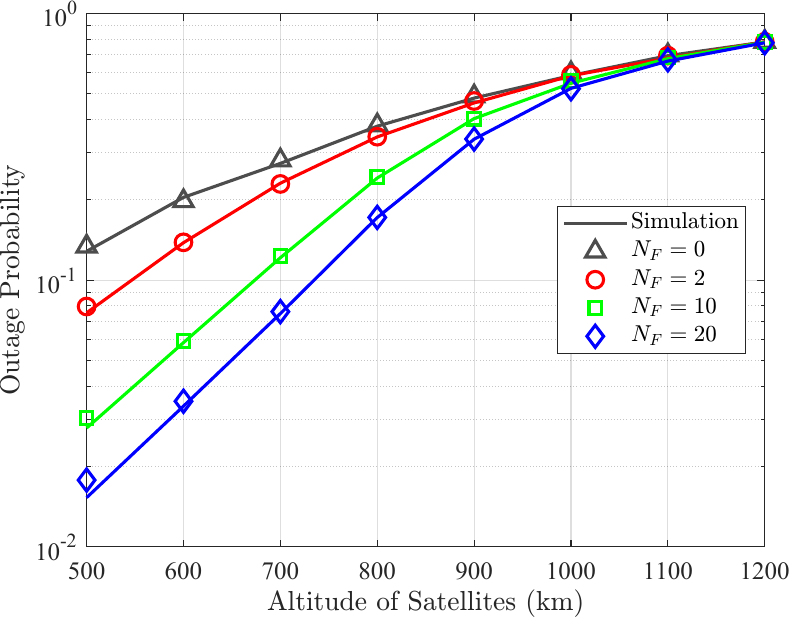}
\caption{Outage probabilities with different altitudes and numbers of follower satellites.}
\label{fig:outage_vs_hsat_log}
\end{figure}

\par
As shown in Fig.~\ref{fig:outage_vs_gamma_log}, at the default altitude and a coverage threshold of $-6$~dB, a cluster with $10$ follower satellites can reduce the outage probability to approximately one-tenth of that of a single leader satellite. In Fig.~\ref{fig:outage_vs_hsat_log}, at the default coverage threshold and an altitude of $500$~km, deploying $20$ follower satellites can reduce the outage probability to $10\%$ of its original value without followers. Furthermore, whether due to the increased path loss from raising satellite altitude or directly increasing the coverage threshold, both scenarios make it more challenging to reach the coverage threshold, consequently increasing the outage probability. An interesting phenomenon is the convergence of the curves in Fig.~\ref{fig:outage_vs_gamma_log} and Fig.~\ref{fig:outage_vs_hsat_log}. This indicates that at high altitudes and high coverage thresholds, follower satellites often experience outages due to their lower transmit power compared to leader satellites. In this case, the communication performance is primarily determined by the leader satellites, and the introduction of followers has a limited impact on communication performance.

\begin{figure}[ht]
\centering
\includegraphics[width=0.9\linewidth]{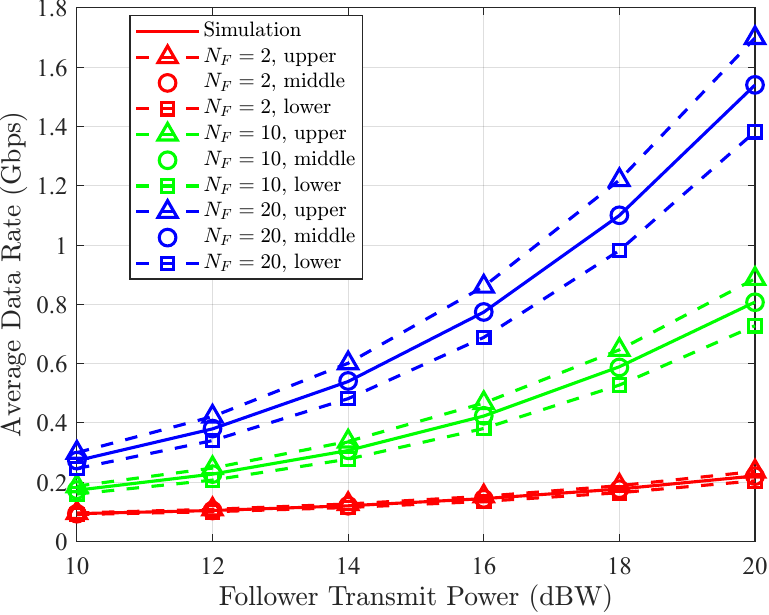}
\caption{Average data rate with different numbers of followers and follower transmit power.}
\label{fig:rate_vs_rho}
\end{figure}
% \begin{figure}[ht]  % 没画middle的
% \centering
% \includegraphics[width=\linewidth]{Fig/rate_vs_rho.png}
% \caption{.}
% \label{fig:rate_vs_rho}
% \end{figure}

\subsection{Average Data Rate}
This subsection shows the impact of follower satellites' configuration on the average data rate. Due to the high computational complexity of the expression in Theorem~\ref{theorem4}, verifying its accuracy is time-consuming. Fig.~\ref{fig:rate_vs_rho} and Fig.~\ref{fig:rate_vs_h_sat} validate that Corollary~\ref{corollary2} can provide relatively tight upper and lower bounds for the average data rate in most cases. In particular, when $N_F$ is small, the network performance is primarily determined by the leader satellite, and the analytical results can provide tight upper and lower bounds. However, as $N_F$ increases and the network performance becomes dominated by the follower satellites, there is a notable deviation between the estimated values from Corollary~\ref{corollary2} and the actual values obtained from simulations. 
% Therefore, we take the average of the upper bound and lower bound estimates as the "middle" label in Fig.~\ref{fig:rate_vs_rho} and Fig.~\ref{fig:rate_vs_h_sat} to obtain an approximate estimation. It can be observed that the values obtained from "middle" exhibit a high degree of overlap with the simulation results.
To mitigate this mismatch, we introduce an approximate estimation obtained by averaging the corresponding upper and lower bounds, which is given by
\begin{equation}\label{eq:mid_rate}
\mathcal{R}_{\mathrm{LFs}}^{\mathrm{middle}}=\frac{1}{2}(\mathcal{R}_{\mathrm{LFs}}^{\mathrm{upper}}+\mathcal{R}_{\mathrm{LFs}}^{\mathrm{lower}}).
\end{equation}
The approximation in \eqref{eq:mid_rate} is illustrated in Fig.~\ref{fig:rate_vs_rho} and Fig.~\ref{fig:rate_vs_h_sat}. It can be observed that the results obtained from \eqref{eq:mid_rate} exhibit a high degree of consistency with the simulation results, especially in the high-rate region.

\begin{figure}[ht]
\centering
\includegraphics[width=0.9\linewidth]{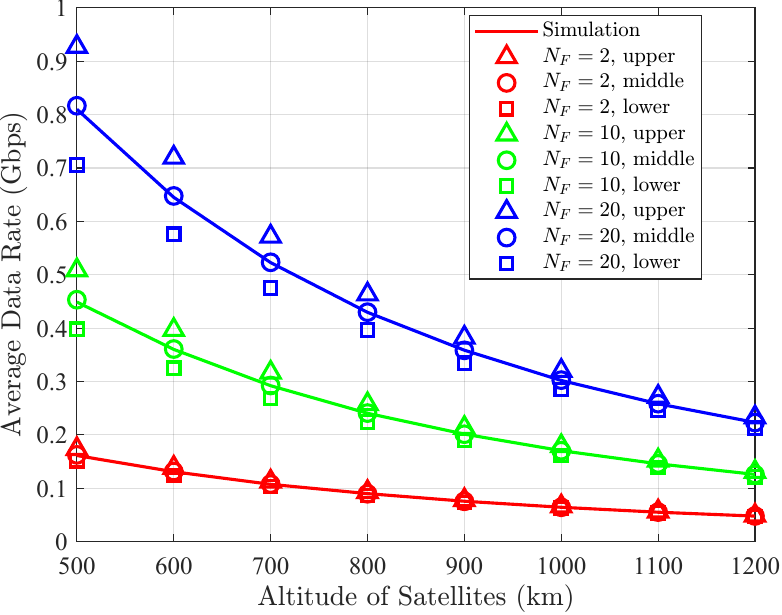}
\caption{Average data rate with different numbers of followers and altitudes.}
\label{fig:rate_vs_h_sat}
\end{figure}

\par
Similar to the convergence effect in outage probability, the improvement in average data rate by follower satellites is more pronounced when the communication quality is relatively high. As shown in Fig.~\ref{fig:rate_vs_rho} and Fig.~\ref{fig:rate_vs_h_sat}, when the follower transmit power is high and the satellite altitude is low, having a cluster with $20$ followers compared to $2$ can result in over five times increase in average data rate. A more intriguing topic is how much improvement in average data rate can be achieved by introducing a leader-follower structure compared to a single leader satellite.

\subsection{Performance Comparison}
To compare the average data rate performance between the leader-follower structure (denoted as "L.F" in the label of Fig.~\ref{fig:rate_N_F} and  Fig.~\ref{fig:rho_rho_LF}) and the non-follower structure (denoted as "N.F" in the label of Fig.~\ref{fig:rate_N_F} and Fig.~\ref{fig:rho_rho_LF}),  we designed the following case study. Considering that the average data rate is typically analyzed in real-time transmission scenarios, in the leader-follower structure, the leader satellite needs to communicate with multiple followers and the user simultaneously. This can result in a lower power allocation from the leader to users compared to a non-follower structure. Additionally, since the distance between the leader and follower satellites is much smaller than the distance between the leader and the user, only a small portion of the transmit power from the leader satellite needs to be used for relay communication with the followers, allowing most of the power to be allocated to the leader-user direct link.

\par
We denote $\rho_{\mathrm{LU}}^{(1)}$ and $\rho_{\mathrm{LU}}^{(2)}$ as the leader-user link's transmit power for non-follower and leader-follower structures, respectively. $\rho_{\mathrm{LF},i}$ is the transmit power from the leader satellite to the $i$-th follower satellite. The following equation is assumed to be satisfied:
\begin{equation}
    \rho_{\mathrm{LU}}^{(1)} = \rho_{\mathrm{LU}}^{(2)} + \sum_{i=1}^{N_F} \rho_{\mathrm{LF},i}.
\end{equation}
In the above expression, the value of $\rho_{\mathrm{LU}}^{(1)}$ is fixed. The value of $\rho_{\mathrm{LF},i}$ is set such that the average data rate from the leader to the $i$-th follower is exactly equal to the average data rate from the $i$-th follower to the user. Therefore, determining the value of $\rho_{\mathrm{LF},i}$ for all $1 \leq i \leq N_\mathrm{F}$ through simulation is also time-consuming. As the average analytical values of the upper and lower bounds provided in Corollary~\ref{corollary2} have been verified to closely match the values obtained through simulation, all numerical results in the case study are derived from this analytical expression. 

\begin{figure}[ht]
\centering
\includegraphics[width=\linewidth]{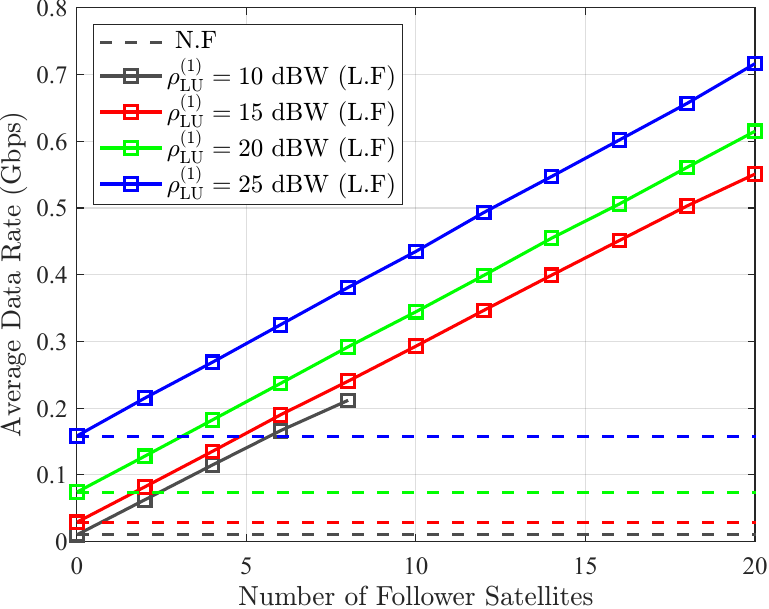}
\caption{Average data rate with different numbers of followers and leader transmit power.}
\label{fig:rate_N_F}
\end{figure}

\par
As shown in Fig.~\ref{fig:rate_N_F}, the average data rate exhibits an approximate linear growth relationship with the increasing number of follower satellites. The starting point of the line is determined by the total transmission power of the leader, while the slope is a constant value due to the fixed transmit power of the followers. In addition, a cluster with $20$ follower satellites can achieve an average data rate that is five times higher than that of a single leader satellite, which demonstrates the effectiveness of the leader-follower structure in expanding satellite channel capacity.

\begin{figure}[ht]
\centering
\includegraphics[width=\linewidth]{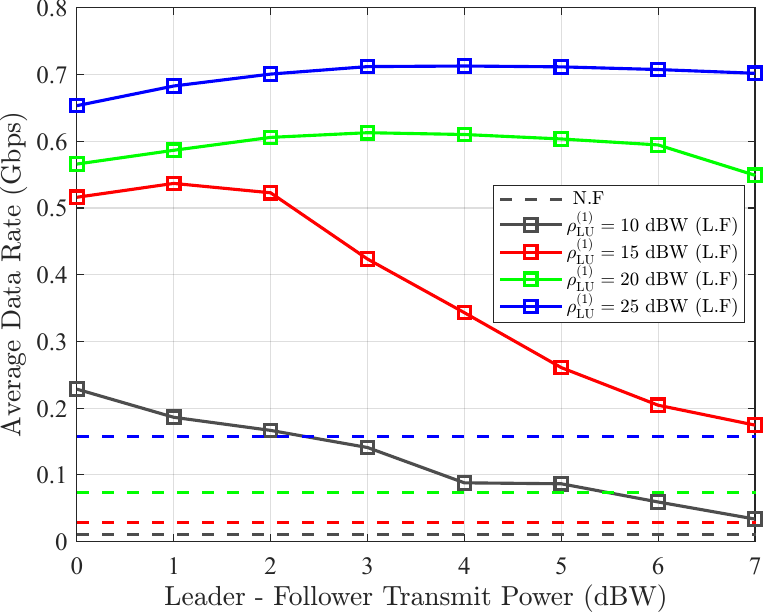}
\caption{Average data rate with different numbers of followers and leader transmit power.}
\label{fig:rho_rho_LF}
\end{figure}

\par
The x-axis of Fig.~\ref{fig:rho_rho_LF} represents the transmit power $\rho_{\mathrm{LF},i}$ allocated by the leader to a single follower satellite. It can be observed that when $\rho_{\mathrm{LU}}^{(1)} \geq 15$~dBW, there exists an optimal transmit power to maximize the average data rate. If $\rho_{\mathrm{LF},i}$ is designed as a small value, the follower satellites can not receive enough data from the leader, thus limiting the overall data rate. On the other hand, if $\rho_{\mathrm{LF},i}$ is large, it may reduce the direct communication power between the leader satellite and the user. Therefore, adjusting the power allocation strategy of the leader satellite based on factors such as the transmit power of the leader and follower satellites, the distances between them, can be a meaningful direction for future research.

\par
Finally, two points are worth mentioning. Firstly, for the black lines in Fig.~\ref{fig:rate_N_F} and Fig.~\ref{fig:rho_rho_LF}, even when all power is allocated to the inter-satellite links, the transmit power $\rho_{\mathrm{LU}}^{(1)} = 10$~dBW is insufficient. Therefore, we chose to reduce the number of followers in the cluster while still allocating transmission power at default values. This adjustment in the power allocation strategy resulted in the premature termination of the black line in Fig.~\ref{fig:rate_N_F} when assembling $8$ follower satellites, which also led to the appearance of non-smooth points in the black line in Fig.~\ref{fig:rho_rho_LF}. Secondly, our analysis of outage probability is based on scenarios with poor transmission conditions. In such cases where the probability of communication interruptions is high, followers can store information from the leader and repeatedly transmit it to users to ensure successful communication. With the offline store-and-forward transmission mechanism, the power allocation scheme was not required to be considered when evaluating the outage performance for both structures.

\section{Further Discussion}
In the previous sections, we discussed the advantages of introducing the leader-follower architecture and quantitatively demonstrated the performance improvements brought by increasing the number of follower satellites in a cluster. Beyond further elaborating on these advantages, this section investigates the potential drawbacks for large clusters as well as the associated challenges. 

\subsection{Advantages for Network Performance}
In this section, we first demonstrate, through numerical simulations, how the number of leader and follower satellites improves the network performance, and then explain the underlying reasons from the perspective of the analytical expressions.

\begin{figure}[ht]
\centering
\includegraphics[width=0.9\linewidth]{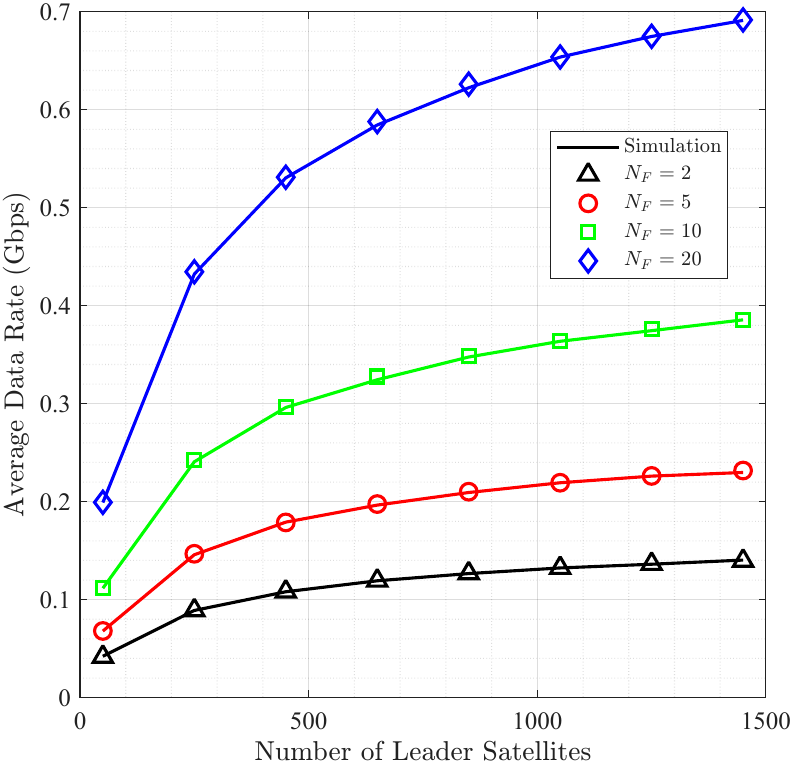}
\caption{Average data rate with different numbers of leader and follower satellites.}
\label{fig:rate_N_L}
\end{figure}

\par
As shown in Fig.~\ref{fig:rate_N_L}, increasing the number of satellites-whether follower satellites or leader satellites-can increase the average data rate. When the number of leader satellites is fixed, increasing the number of follower satellites can always effectively improve the average data rate. Conversely, when the number of followers in a cluster is fixed, once the number of leader satellites reaches a certain value, further increasing the deployment density of leader satellites yields limited performance improvement.

\par
From the analytical expressions, as the number of satellites increases, the value of the contact angle CDF increases. In other words, with a larger number of satellites, the PDF of the contact angle is higher at small contact angle values, leading to a lower outage probability and a higher average data rate. 

\subsection{Disadvantages for Network Performance}
When communication is continuous, and the data volume is large, the transmission delay required for packet delivery becomes the dominant factor determining latency. The data rate is determined by the smaller of the rates in the leader-follower link and the follower-user link. In this case, increasing the number of follower satellites leads to a continuous improvement in the network's average data rate. However, if communication is bursty and the data volume is small, the information may reach the user in a single transmission. In this scenario, transmission delay is no longer dominant and can even be shorter than the propagation delay through the atmosphere. Using a follower as a relay in such cases introduces not only additional propagation delay but also demodulation-forwarding delay and queuing delay.

\par
The second drawback of a large cluster is the significantly increased co-channel interference. While inter-satellite communication between leader and follower satellites can reduce intra-cluster interference through coordination, interference between clusters will increase substantially. Similarly, for ground users,  interference from signals within the same cluster can be easily mitigated, but the interference from other clusters will be significantly stronger. In summary, large clusters are suitable only for specific scenarios. When a single satellite handles a large volume of data traffic, introducing follower satellites to assist in transmission can effectively improve the data rate. However, in lightly loaded satellite networks, the role of followers is limited. Moreover, in already densely deployed satellite networks, adding more follower satellites may introduce additional inter-cluster interference.

\subsection{Challenges of Large Clusters}
Large clusters also introduce three major challenges. First is the higher cost, including the construction, launch, and recovery of follower satellites. In addition, the leader satellite must be equipped with more transmitters to maintain continuous communication with the followers within the cluster. Second is the increased collision risk. Some configurations require follower satellites to be concentrated within a small spherical cap, raising concerns about collisions between densely packed, fast-moving satellites, as well as collisions with space debris. Furthermore, when the orbits of two clusters intersect, avoiding collisions becomes a significant challenge. Third is the limitation caused by high centralization. Since follower satellites rely on the leader for scheduling and are dedicated solely to communication services, large clusters heavily depend on the proper operation of the leader. If the leader fails to provide communication services, the followers cannot receive instructions and therefore cannot serve the users.

\subsection{Summary}
As a result, providing a deterministic answer for the number of follower satellites is challenging. When selecting this number, designers may take the following factors into consideration:
\begin{itemize}
\item The practical data rate requirements of the network and the allowable outage probability. 
\item The intrinsic cost of the follower satellites, as well as the additional cost incurred by equipping leader satellites with multiple extra transmitters. 
\item Interference between leader and follower satellites from different clusters.
\item The risk of collisions between satellites within a cluster, between clusters, or with space debris.
\item Centralization leads to follower satellites being inoperative if the leader satellite experiences a communication failure.
\end{itemize}

\par
By jointly considering all these factors, the trade-off in the number of satellites within a cluster can be clearly characterized.

\section{Conclusion}
In this paper, we developed a low-complexity performance analysis framework for the leader-follower satellite system for the first time. The analytical results derived from this framework not only provide tight upper and lower bounds for outage probability and average data rate in most cases but also offer accurate performance estimates for the system. Numerical results indicate that a satellite cluster with multiple followers can significantly enhance the performance of the satellite network compared to a single leader satellite. In some configurations, the outage probability can be reduced to one-tenth of the original, while the average data rate can increase by more than five times the original value. The results indicate both the spatial diversity advantages of satellite clusters and the effectiveness of follower structures in extending individual satellites' service capabilities. 

\par
To compare the real-time average data rate performance with and without follower satellites, we designed a case study and introduced a power allocation mechanism for the leader satellite. Since power allocation through simulation is computationally expensive, the case study was entirely conducted using analytical results. Therefore, the case study can be regarded as an important application of the provided analytical framework. The numerical results also yield two significant conclusions. As follower satellites are configured to be fewer than leader satellites, improving communication quality by increasing the number of follower satellites relies on the premise that follower satellites can maintain high-quality and reliable communication with users. Furthermore, there exists an optimal power allocation strategy, making it crucial to determine the power allocation from the leader to different followers. Therefore, the design of allocation strategies can be considered a future research direction.

\appendices
\section{Proof of Lemma~\ref{lemma1}}\label{app:lemma1}
Based on Definition~\ref{definition2}, the CDF of the contact angle between the associated leader satellite and the user can be derived as:
\begin{equation} 
    \begin{split}
        F_{\theta_{\mathrm{LU}}}(\theta) &= 1 - \mathbb{P}\left[N\big(\mathcal{A}_{\mathrm{cap}}(\theta)\big) = 0\right],
    \end{split}
\end{equation}
where $\mathcal{A}_{\mathrm{cap}}(\theta)$ denotes the area of a spherical cap with central angle $2\theta$, and $N(\mathcal{A}_{\mathrm{cap}}(\theta))$ counts the number of leader satellites located within this cap region.

\par
Based on the homogeneous BPP assumption for satellite distribution, the probability that no satellite is located within the cap is
\begin{equation}
\begin{split}
& \mathbb{P}\left[N\big(\mathcal{A}_{\mathrm{cap}}(\theta)\big) = 0\right] \\
&= \left( 1 - \frac{2\pi R_{\mathrm{sat}}^{2}(1 - \cos\theta)}{4\pi R_{\mathrm{sat}}^{2}} \right)^{N_{\mathrm{L}}} \\
&= \left( \frac{1 + \cos\theta}{2} \right)^{N_{\mathrm{L}}}.
    \end{split}
\end{equation}
Substituting it into the CDF expression yields
\begin{equation}\label{app:F_LU}
    F_{\theta_{\mathrm{LU}}}(\theta) = 1 - \left( \frac{1 + \cos\theta}{2} \right)^{N_{\mathrm{L}}}.
\end{equation}

\par
Taking the derivative of the CDF with respect to $\theta$ leads to the PDF:
\begin{equation}
    f_{\theta_{\mathrm{LU}}}(\theta) = \frac{N_{\mathrm{L}} \sin\theta}{2}  \left( \frac{1 + \cos\theta}{2} \right)^{N_{\mathrm{L}} - 1}.
\end{equation}

\par
Regarding the range of the contact angle, the minimum value of $\theta_{\mathrm{LU}}$ is achieved when a leader satellite is located directly above the user, while the maximum value corresponds to the case where the leader lies at the farthest allowable distance from the user, denoted by $\theta_{\mathrm{max}}$.

\section{Proof of Lemma~\ref{lemma2}}\label{app:lemma2}
For the convenience of subsequent derivations, we rotate the spherical coordinate system and set the coordinates of the associated leader satellite as located at $(R_{\mathrm{sat}}, 0, 0)$. In this case, the user is located at $(R_{\oplus}, \theta_{\mathrm{LU}}, 0)$. The follower satellites are located on a spherical cap with the $z$-axis as the rotation axis. Next, we discuss two cases based on the value of $\theta_{\mathrm{LU}}$.

\par
\begin{figure}[ht]
\centering
\includegraphics[width=\linewidth]{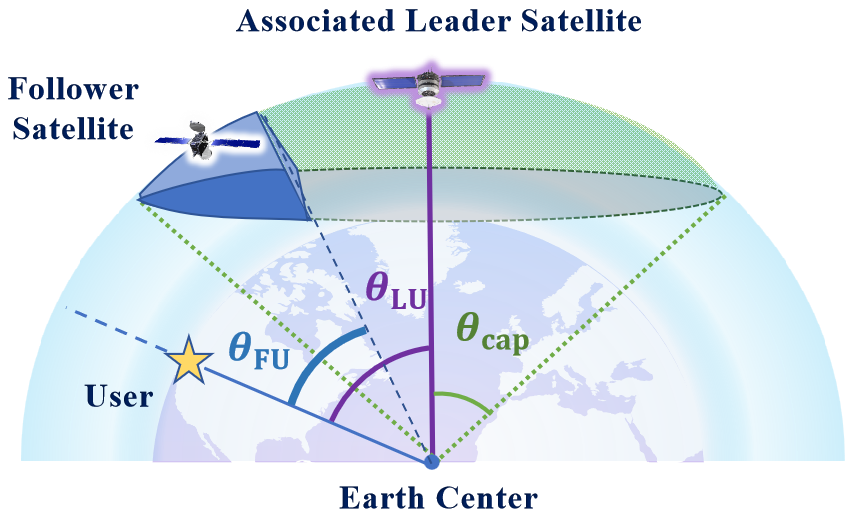}
\caption{Geometric illustration for $\theta_{\mathrm{FU}}$ with the user outside the spherical cap.}
\label{figure2}
\end{figure}
In the first case, $\theta_{\mathrm{LU}} \in [\theta_{\mathrm{cap}},\theta_{\mathrm{max}}]$ is satisfied, and the user is located outside the spherical cap where the followers are distributed. Since the followers follow a BPP within the cap, the CDF of $\theta_{\mathrm{FU}}$ can be expressed as
\begin{equation}
F_{\theta_{\mathrm{FU}}}\left(\theta\right)=\frac{S_{\mathrm{shadow}}}{S_{\mathrm{cap}}},
\end{equation}
where $S_{\mathrm{shadow}}$ denotes the surface area of the dark blue shaded region shown in Fig.~\ref{figure2}. $S_{\mathrm{cap}} = 2\pi R_{\mathrm{sat}}^{2}(1 - \cos\theta_{\mathrm{cap}})$ is the surface area of the full spherical cap.

\par
The surface area of an intersected portion on one side of a spherical cap with central angle $\theta_{\mathrm{c}}$, and intercepted by angle $\theta_{\mathrm{o}}$, is given by \cite{wang2022conditional}:
\begin{equation}
\begin{split}
    S(\theta_{\mathrm{c}}, \theta_{\mathrm{o}}) &= \int_{R_{\mathrm{sat}} \cos \theta_{\mathrm{c}} \tan \left( \theta_{\mathrm{c}} - \theta_{\mathrm{o}} \right)}^{R_{\mathrm{sat}} \sin \theta_{\mathrm{c}}} 2 R_{\mathrm{sat}} \\
    &\quad \times \arcsin \left( \frac{\sqrt{R_{\mathrm{sat}}^2 \sin^2 \theta_{\mathrm{c}} - l^2}}{R_{\mathrm{sat}}} \right) \, dl.
\end{split}
\end{equation}
Based on the angular relationship illustrated in the figure, the dark blue shaded region corresponds to the intersected portion of the spherical cap, with the corresponding angular extent given by $\theta_{\mathrm{cap}} - \theta_{\mathrm{LU}} + \theta$. Thus, the corresponding CDF becomes
\begin{equation}
    F_{\theta_{\mathrm{FU}}}(\theta) = \frac{S\left(\theta_{\mathrm{cap}},\theta_{\mathrm{cap}}-\theta_{\mathrm{LU}}+\theta\right)}{2\pi R_{\mathrm{sat}}^{2}(1 - \cos\theta_{\mathrm{cap}})}.
\end{equation}

The PDF of $\theta_{\mathrm{FU}}$ can be obtained by differentiating its CDF with respect to $\theta$. To facilitate the derivation, we compute the derivative of the function $S(\theta_{\mathrm{c}}, \theta_{\mathrm{o}})$ with respect to $\theta_{\mathrm{o}}$. Since the upper limit of the integral is independent of $\theta_{\mathrm{o}}$, and the integrand does not contain $\theta_{\mathrm{o}}$ explicitly, the derivative can be computed using the Leibniz integral rule:
\begin{equation}
    \frac{\partial S(\theta_{\mathrm{c}}, \theta_{\mathrm{o}})}{\partial \theta_{\mathrm{o}}}
    = - f(l_0, \theta_{\mathrm{c}}) \cdot \frac{\partial l_0}{\partial \theta_{\mathrm{o}}},
\end{equation}
where the lower limit $l_0$ is given by
\begin{equation}
    l_0 = R_{\mathrm{sat}} \cos \theta_{\mathrm{c}} \tan(\theta_{\mathrm{c}} - \theta_{\mathrm{o}}).
\end{equation}
Differentiating $l_0$ with respect to $\theta_{\mathrm{o}}$ yields:
\begin{equation}
    \frac{\partial l_0}{\partial \theta_{\mathrm{o}}} 
    = - R_{\mathrm{sat}} \cos \theta_{\mathrm{c}} \cdot \sec^2(\theta_{\mathrm{c}} - \theta_{\mathrm{o}}).
\end{equation}
Evaluating the integrand at the lower limit $l = l_0$, we obtain:
\begin{equation}
\begin{aligned}
    f(l_0, \theta_{\mathrm{c}}) &= 2 R_{\mathrm{sat}} \arcsin \left( 
    \sqrt{ \sin^2 \theta_{\mathrm{c}} - \left( \frac{l_0}{R_{\mathrm{sat}}} \right)^2 } \right) \\
    &= 2 R_{\mathrm{sat}} \arcsin \left( 
    \sqrt{ \sin^2 \theta_{\mathrm{c}} - \cos^2 \theta_{\mathrm{c}} \tan^2(\theta_{\mathrm{c}} - \theta_{\mathrm{o}}) } \right).
\end{aligned}
\end{equation}
Substituting the above expressions into the Leibniz formula, we obtain the final derivative:
\begin{equation}
\begin{aligned}
    \frac{\partial S(\theta_{\mathrm{c}}, \theta_{\mathrm{o}})}{\partial \theta_{\mathrm{o}}}
    &= 2 R_{\mathrm{sat}}^2 \cos \theta_{\mathrm{c}} \sec^2(\theta_{\mathrm{c}} - \theta_{\mathrm{o}}) \\
    &\quad \times \arcsin \left( \sqrt{ \sin^2 \theta_{\mathrm{c}} - \cos^2 \theta_{\mathrm{c}} \tan^2(\theta_{\mathrm{c}} - \theta_{\mathrm{o}}) } \right).
\end{aligned}
\end{equation}
Finally, by applying the derivative of the function $S(\theta_{\mathrm{c}}, \theta_{\mathrm{o}})$, the PDF can be obtained by differentiating the CDF as follows:
\begin{equation}
\begin{split}  
    & f_{\theta_{\mathrm{FU}}}(\theta) = \frac{ \cos \theta_{\mathrm{cap}} \sec^2(\theta_{\mathrm{LU}} - \theta) }{\pi (1 - \cos\theta_{\mathrm{cap}})}\\
     &\times\arcsin \left( \sqrt{ \sin^2 \theta_{\mathrm{cap}} - \cos^2 \theta_{\mathrm{cap}} \tan^2(\theta_{\mathrm{LU}} - \theta) } \right).
\end{split}
\end{equation}

\par
\begin{figure*}[ht]
\centering
\includegraphics[width=\linewidth]{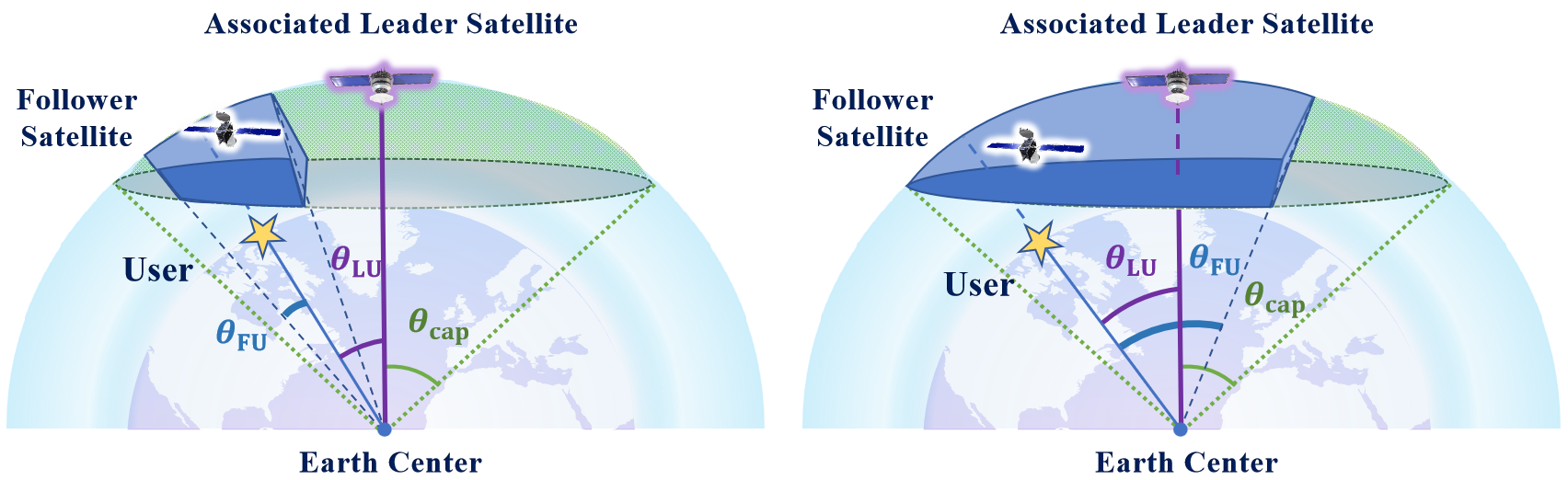}
\caption{Geometric illustration for $\theta_{\mathrm{FU}}$ with the user inside the spherical cap.}
\label{figure3}
\end{figure*}
In the second case, $\theta_{\mathrm{LU}} \in [0,\theta_{\mathrm{cap}})$ is satisfied, and the user is located beneath the spherical cap. The CDF of $\theta_{\mathrm{FU}}$ is again equal to the probability that a follower lies in the dark blue shaded region, and its expression becomes piecewise. 

When $\theta_{\mathrm{FU}}\in[0,\theta_{\mathrm{cap}}-\theta_{\mathrm{LU}}]$, the dark blue shaded region can be computed as the difference between two spherical cap segments, as illustrated in the left part of Fig.~\ref{figure3}. These segments are defined by the angular distances from the edge of the cap to the left and right boundaries of the shaded region, respectively. The corresponding area can thus be calculated as follows:
\begin{equation}
\begin{split}
S_{\mathrm{shadow}}=
S\left(\theta_{\mathrm{cap}},\theta_{\mathrm{cap}}-\theta_{\mathrm{LU}}+\theta\right)\\
-S\left(\theta_{\mathrm{cap}},\theta_{\mathrm{cap}}-\theta_{\mathrm{LU}}-\theta\right).
\end{split}
\end{equation}

When $\theta_{\mathrm{FU}} \in [\theta_{\mathrm{cap}} - \theta_{\mathrm{LU}}, \theta_{\mathrm{cap}} + \theta_{\mathrm{LU}}]$, the corresponding dark blue shaded region is illustrated in the right part of Fig.~\ref{figure3}. The shaded area remains geometrically equivalent to the first case.

As a result, the CDF of $\theta_{\mathrm{FU}}$ is given as
\begin{equation}
\begin{split}
& F_{\theta_{\mathrm{FU}}}(\theta) = \\
&\begin{cases}
\dfrac{S\left(\theta_{\mathrm{cap}},\theta_{\mathrm{cap}}-\theta_{\mathrm{LU}}+\theta\right)-S\left(\theta_{\mathrm{cap}},\theta_{\mathrm{cap}}-\theta_{\mathrm{LU}}-\theta\right)}{2\pi R_{\mathrm{sat}}^{2}(1 - \cos\theta_{\mathrm{cap}})},\\
\quad \quad \quad \quad \quad \quad \quad \quad \quad \quad \quad \quad \text{for } \theta\in[0,\theta_{\mathrm{cap}}-\theta_{\mathrm{LU}}], \\
\dfrac{S\left(\theta_{\mathrm{cap}},\theta_{\mathrm{cap}}-\theta_{\mathrm{LU}}+\theta\right)}{2\pi R_{\mathrm{sat}}^{2}(1 - \cos\theta_{\mathrm{cap}})}, \\
\quad \quad \quad \quad \quad \quad \quad \quad \text{for } \theta\in[\theta_{\mathrm{cap}}-\theta_{\mathrm{LU}},\theta_{\mathrm{cap}}+\theta_{\mathrm{LU}}],
\end{cases}
\end{split}
\end{equation}
and the corresponding PDF can be obtained by differentiating the piecewise CDF with respect to $\theta$.

\section{Proof of Lemma~\ref{lemma3}}\label{app:lemma3}
In this appendix, we first derive the CDF and PDF of the distributions of minimum contact angle $\theta_{\mathrm{FU}}^{\mathrm{min}}$. Following the same coordinate system transformation approach as in Appendix~\ref{app:lemma2}, we consider two distinct cases: (i) the user is located beneath the spherical cap ($\theta_{\mathrm{LU}}\leq\theta_{\mathrm{cap}}$) where follower satellites are distributed and (ii) the user is not ($\theta_{\mathrm{LU}}>\theta_{\mathrm{cap}}$).

\par
For the first case, the minimum contact angle is considered as the condition that the nearest follower is located directly above the user, that is,  $\theta_{\mathrm{FU}}^{\mathrm{min}}=0$. Consequently, the probability mass at $\theta_{\mathrm{FU}}^{\mathrm{min}}=0$ is equal to the value  $F_{\theta_{\mathrm{LU}}}(\theta_{\mathrm{cap}})$. Therefore, the PDF of $f_{\theta_{\mathrm{LU}}}(\theta_{\mathrm{cap}})$ can be written as
\begin{equation}
    f_{\theta_{\mathrm{FU}}^{\mathrm{min}}}(\theta) = \left(1 - \left( \dfrac{1 + \cos(\theta_{\mathrm{cap}})}{2} \right)^{N_{\mathrm{L}}}\right)\cdot \delta(\theta),
\end{equation}
when $\theta = 0$. $\delta(\cdot)$ denotes the Dirac delta function, which accounts for the discrete probability mass at $\theta = 0$.

\par
Then, maximum contact angle corresponds to the edge of the spherical cap, yielding $\theta_{\mathrm{FU}}^{\mathrm{min}}=\theta_{\mathrm{LU}}-\theta_{\mathrm{cap}}$. The CDF is derived as follows:
\begin{equation}
\begin{split}
    F_{\theta_{\mathrm{FU}}^{\mathrm{min}}}(\theta)
    &= \mathbb{P}(\theta_{\mathrm{FU}}^{\mathrm{min}} \leq \theta)\\
    &= \mathbb{P}(\theta_{\mathrm{LU}} \leq \theta + \theta_{\mathrm{cap}}) \\
    &= F_{\theta_{\mathrm{LU}}}(\theta + \theta_{\mathrm{cap}}).
    \end{split}
\end{equation}
Substituting from (\ref{app:F_LU}), we obtain the final CDF expression:
\begin{equation}
    F_{\theta_{\mathrm{FU}}^{\mathrm{min}}}(\theta)= 1 - \left( \frac{1 + \cos(\theta + \theta_{\mathrm{cap}})}{2} \right)^{N_{\mathrm{L}}}.
\end{equation}

\par
The PDF is obtained by differentiating the CDF with respect to $\theta$, yielding:
\begin{equation}
\begin{split}
f_{\theta_{\mathrm{FU}}^{\mathrm{min}}}(\theta) = \frac{N_{\mathrm{L}} \sin(\theta + \theta_{\mathrm{cap}})}{2} \left( \dfrac{1 + \cos(\theta + \theta_{\mathrm{cap}})}{2} \right)^{N_{\mathrm{L}} - 1}. 
\end{split}
\end{equation}

\par
Next, we derive the CDF and the PDF of the distributions of maximum contact angle $\theta_{\mathrm{FU}}^{\max}$. It can be observed that the farthest follower lies at the boundary of the spherical cap in the direction opposite to the user, resulting in the maximum contact angle
\begin{equation}
    \theta_{\mathrm{FU}}^{\mathrm{max}} = \theta_{\mathrm{LU}} + \theta_{\mathrm{cap}}.
\end{equation}
Accordingly, the CDF of $\theta_{\mathrm{FU}}^{\mathrm{max}}$ is given by
\begin{equation}
\begin{split}
    F_{\theta_{\mathrm{FU}}^{\mathrm{max}}}(\theta)
    &= \mathbb{P}(\theta_{\mathrm{FU}}^{\mathrm{max}} \leq \theta) \\
    &= \mathbb{P}(\theta_{\mathrm{LU}} \leq \theta - \theta_{\mathrm{cap}}) \\
    &= F_{\theta_{\mathrm{LU}}}(\theta - \theta_{\mathrm{cap}}).
\end{split}
\end{equation}

\par
Following a similar approach as in the previous derivation, the PDF of $\theta_{\mathrm{FU}}^{\mathrm{max}}$ can be obtained by first substituting the CDF expression and then differentiating it with respect to $\theta$, yielding
\begin{equation}
    f_{\theta_{\mathrm{FU}}^{\mathrm{max}}}(\theta) =
    \dfrac{N_{\mathrm{L}} \sin(\theta - \theta_{\mathrm{cap}})}{2} 
    \left( \dfrac{1 + \cos(\theta - \theta_{\mathrm{cap}})}{2} \right)^{N_{\mathrm{L}} - 1}.
\end{equation}
Since $\theta_{\mathrm{LU}} \in [0, \theta_{\max}]$, the maximum contact angle $\theta_{\mathrm{FU}}^{\mathrm{max}}$ is bounded within $\left[\theta_{\mathrm{cap}}, \theta_{\max} + \theta_{\mathrm{cap}} \right]$.

\section{Proof of Theorem~\ref{theorem1}}\label{app:theorem1}
As defined in Definition~\ref{definition3}, the outage probability of the associated leader satellite to the user link can be given as 
\begin{equation}\label{P_out_LU}
\begin{split}
& P_{\mathrm{out}}^{\mathrm{LU}} = \mathbb{P} \{ \mathrm{SNR_{\mathrm{LU}}}< \gamma_{_{\mathrm{th}}} \} \\
&=\mathbb{E}_{r_{\mathrm{LU}}}\left[\mathbb{P}\left(\rho_{\mathrm{LU}}G_{\mathrm{LU}}\zeta_{\mathrm{U}}W_{\mathrm{U}}\left(\frac{\nu_{\mathrm{LU}}}{4\pi\sigma_{\mathrm{U}}r_{\mathrm{LU}}}\right)^2<\gamma_{\mathrm{th}}\right)\right]\\
&=\mathbb{E}_{r_{\mathrm{LU}}}\left[\mathbb{P}\left(W_{\mathrm{U}}<\frac{\gamma_{\mathrm{th}}}{\rho_{\mathrm{LU}}G_{\mathrm{LU}}\zeta_{\mathrm{U}}} \left(\frac{4\pi\sigma_{\mathrm{U}}r_{\mathrm{LU}}}{\nu_{\mathrm{LU}}}\right)^2\right)\right],
    \end{split}
\end{equation}
where $r_{\mathrm{LU}}$ is the Euclidean distance between the leader satellite and the user,
\begin{equation}\label{app_r_LU}
    r_{\mathrm{LU}}(\theta_{\mathrm{LU}}) = \sqrt{R_{\mathrm{sat}}^2 + R_{\oplus}^2 - 2 R_{\mathrm{sat}} R_{\oplus} \cos\theta_{\mathrm{LU}}}.
\end{equation}
By substituting the expression of $r_{\mathrm{LU}}(\theta_{\mathrm{LU}})$ and the PDF of $\theta_{\mathrm{LU}}$ into the expectation, the outage probability can be expressed as
\begin{sequation}
    \begin{split}
    P_{\mathrm{out}}^{\mathrm{LU}} = \int_{0}^{\theta_{\mathrm{max}}} \mathbb{P} \left(W_{\mathrm{U}}<\frac{\gamma_{\mathrm{th}}}{\rho_{\mathrm{LU}}G_{\mathrm{LU}}\zeta_{\mathrm{U}}} \left(\frac{4\pi\sigma_{\mathrm{U}}r_{\mathrm{LU}}(\theta)}{\nu_{\mathrm{LU}}}\right)^2\right) \\
    \times f_{\theta_{\mathrm{LU}}}(\theta) \, \mathrm{d}\theta \\
    = \int_{0}^{\theta_{\max}} 
    F_{W_{\mathrm{U}}} \left( \frac{\gamma_{\mathrm{th}}}{\rho_{\mathrm{LU}}G_{\mathrm{LU}}\zeta_{\mathrm{U}}}\left(\frac{4\pi\sigma_{\mathrm{U}}r_{\mathrm{LU}}(\theta)}{\nu_{\mathrm{LU}}}\right)^2 \right) \\
    \times f_{\theta_{\mathrm{LU}}}(\theta) \, \mathrm{d} \theta.
    \end{split}
\end{sequation}

% \begin{equation}
%     \begin{split}
%         P_{\mathrm{out}}^{\mathrm{LFs}}= \int_{0}^{\theta_{\max}} F_{W_s} \left( \frac{\gamma_{\mathrm{th}}}{C} \cdot \left( r_{\mathrm{FU}}(\theta_{\mathrm{FU}}) \right)^2 \right) f_{\theta_{\mathrm{FU}}}(\theta_{\mathrm{FU}}) d\theta_{\mathrm{FU}},
%     \end{split}
% \end{equation}
% where  $r_{\mathrm{FU}}(\theta_{\mathrm{FU}})$ denotes the distance between the typical user and the nearest follower satellite, which can be expressed as
% \begin{equation}\label{r_LU}
%     r_{\mathrm{FU}}(\theta_{\mathrm{FU}}) = \sqrt{R_{\mathrm{sat}}^2 + R_{\oplus}^2 - 2 R_{\mathrm{sat}} R_{\oplus} \cos(\theta_{FU})}.
% \end{equation}.

% \begin{equation}
%     P_{\mathrm{out}}^{\mathrm{LFs}} = \mathbb{E}_{\theta_{\mathrm{LU}}}\left[\left(P_{\mathrm{out},i}^{\mathrm{FU}}\left(\theta_{\mathrm{LU}}\right) \right)^{N_{\mathrm{F}}}\right],
% \end{equation}

% \begin{equation}\label{P_out_fu_i}
%     \begin{split}
%         P_{\mathrm{out},i}^{\mathrm{FU}} 
%     =& \int_{0}^{2\pi} \int_{\theta}^{\theta_{\mathrm{cap}}} \int_{0}^{\theta_{\mathrm{max}}} 
%     F_{W_s}\left( 
%         \frac{ \gamma_{th} \sigma^2 }{ \rho_{\mathrm{sat}} G \zeta \nu^2 } 
%         \cdot \left( 4\pi r_{\mathrm{FU}} \right)^2 
%     \right) \notag \\
%     &\times f_{\theta_{\mathrm{LU}}}(\theta_{\mathrm{LU}})
%     f_{\theta_\mathrm{F}}(\theta)
%     f_{\phi_{\mathrm{F}}}(\phi)
%     \, d\theta_{\mathrm{LU}} d\theta d\phi.
%     \end{split}
% \end{equation}

\section{Proof of Theorem~\ref{theorem2}} \label{app:theorem2}
The outage probability of a satellite cluster equipped with follower satellites can be derived by first analyzing the outage probability between a follower satellite and the user. Given the central angle $\theta_{\mathrm{LU}}$, the conditional outage probability between the user and the $i$-th follower satellite can be expressed as:
\begin{sequation}
P_{\mathrm{out},i}^{\mathrm{FU}} = 
\mathbb{E}_{r_{\mathrm{FU}}}\left[\mathbb{P}\left(W_{\mathrm{U}}<\frac{\gamma_{\mathrm{th}}}{\rho_{\mathrm{FU}}G_{\mathrm{FU}}\zeta_{\mathrm{U}}} \left(\frac{4\pi\sigma_{\mathrm{U}}r_{\mathrm{FU}}}{\nu_{\mathrm{FU}}}\right)^2\right)\right],
\end{sequation}
where $r_{\mathrm{FU}}$ denotes the Euclidean distance between the follower satellite and the user:
\begin{equation}\label{app:r_FU}
\begin{aligned}
& r_{\mathrm{FU}}\left(\theta_{\mathrm{LU}}, \psi, \varphi\right) = \Big( R_{\oplus}^2 + R_{\mathrm{sat}}^2 
- 2 R_{\oplus} R_{\mathrm{sat}} \\
& \times \big( 
\sin\theta_{\mathrm{LU}} \sin\psi \cos\varphi + \cos\theta_{\mathrm{LU}} \cos\psi 
\big) \Big)^{1/2}.
\end{aligned}
\end{equation}

\par
Since the follower satellites are distributed according to a BPP over a spherical cap with angular radius $\theta_{\mathrm{cap}}$, the corresponding PDFs of the polar and azimuth angles are given by:
\begin{align}
f_{\psi}(\psi) &= \frac{\sin\psi}{1 - \cos\theta_{\mathrm{cap}}}, \quad 0 \leq \psi \leq \theta_{\mathrm{cap}}, \\
f_{\varphi}(\varphi) &= \frac{1}{2\pi}, \quad 0 \leq \varphi \leq 2\pi.
\end{align}
For a fixed $\theta_{\mathrm{LU}}$, the average outage probability can be obtained by taking the expectation over $\psi$ and $\varphi$:
\begin{sequation}
\begin{split}
& P_{\mathrm{out},i}^{\mathrm{FU}} \left(\theta_{\mathrm{LU}}\right) = \int_{0}^{2\pi} \int_{0}^{\theta_{\mathrm{cap}}} \frac{\sin \psi}{2\pi\left(1 - \cos\theta_{\mathrm{cap}}\right)} \\
& \times F_{W_{\mathrm{U}}} \left( \frac{\gamma_{\mathrm{th}}}{\rho_{\mathrm{FU}}G_{\mathrm{FU}}\zeta_{\mathrm{U}}} \left(\frac{4\pi\sigma_{\mathrm{U}}r_{\mathrm{FU}} \left(\theta, \psi, \varphi \right)}{\nu_{\mathrm{FU}}}\right)^2 \right) \mathrm{d}\psi \mathrm{d}\varphi.
\end{split}
\end{sequation}

\par
Since the fading and positions of follower-user links and the leader-user link are assumed to be independent, the outage events for different satellites are mutually independent. Therefore, the overall outage probability with a given $\theta_{\mathrm{LU}}$ can be expressed as:
\begin{equation}
P_{\mathrm{out,cond}}^{\mathrm{Cluster}}\left(\theta_{\mathrm{LU}}\right) =\left( P_{\mathrm{out},i}^{\mathrm{FU}} \left(\theta_{\mathrm{LU}}\right)\right)^{N_{\mathrm{F}}} \cdot P_{\mathrm{Cond}}^{\mathrm{LU}}\left(\theta_{\mathrm{LU}}\right),
\end{equation}
where $P_{\mathrm{out}}^{\mathrm{Cond}}\left(\theta_{\mathrm{LU}}\right)$ is the outage probability for the associated leader satellite conditioned on $\theta_{\mathrm{LU}}$. Its derivation follows a similar approach to that in equation (\ref{P_out_LU}),
\begin{equation}
\begin{split}
& P_{\mathrm{Cond}}^{\mathrm{LU}}\left(\theta_{\mathrm{LU}} \right) \\
& = F_{W_{\mathrm{U}}} \left( \frac{\gamma_{\mathrm{th}}}{\rho_{\mathrm{LU}}G_{\mathrm{LU}}\zeta_{\mathrm{U}}} \left(\frac{4\pi\sigma_{\mathrm{U}}r_{\mathrm{LU}}(\theta_{\mathrm{LU}})}{\nu_{\mathrm{LU}}}\right)^2 \right).
    \end{split}
\end{equation}

\par
The overall unconditional outage probability for the satellite cluster is obtained by traversing $\theta_{\mathrm{LU}}$:
\begin{equation}
P_{\mathrm{out}}^{\mathrm{Cluster}} =\int_{0}^{\theta_{\mathrm{max}}} \left(P_{\mathrm{out},i}^{\mathrm{FU}}\left(\theta\right) \right)^{N_{\mathrm{F}}}P_{\mathrm{Cond}}^{\mathrm{LU}}\left(\theta\right) f_{\theta_{\mathrm{LU}}}(\theta)  \mathrm{d} \theta.
\end{equation}

% \begin{equation}
%     \begin{split}
%         P_{\mathrm{out}}^{\mathrm{min}}= \int_{0}^{\theta_{\max}} F_{W_s} \left( \frac{\gamma_{\mathrm{th}}}{C} \cdot \left( r_{\mathrm{FU}}(\theta_{\mathrm{FU}}^{\mathrm{min}}) \right)^2 \right) f_{\theta_{\mathrm{FU}}^{\mathrm{min}}}(\theta) d\theta,
%     \end{split}
% \end{equation}
% \begin{equation}
%     \begin{split}
%         P_{\mathrm{out}}^{\mathrm{max}}= \int_{0}^{\theta_{\max}} F_{W_s} \left( \frac{\gamma_{\mathrm{th}}}{C} \cdot \left( r_{\mathrm{FU}}(\theta_{\mathrm{FU}}^{\mathrm{max}}) \right)^2 \right) f_{\theta_{\mathrm{FU}}^{\mathrm{max}}}(\theta) d\theta,
%     \end{split}
% \end{equation}
% where $f_{\theta_{\mathrm{FU}}^{\mathrm{max}}}(\theta)$ and $f_{\theta_{\mathrm{FU}}^{\mathrm{max}}}(\theta)$ are given in Lemma~\ref{lemma2}.

\section{Proof of Theorem~\ref{theorem3}}\label{app:theorem3}
Based on  Definition~\ref{definition4}, the average data rate of the leader satellite can be computed as
\begin{equation}
\begin{split}
    & \mathcal{R}_{\mathrm{LU}} = \mathbb{E}_{r_{\mathrm{LU}}, W_{\mathrm{U}}} \Big[ B_{\mathrm{LU}} \\
    & \times \log_2\left(1 + \rho_{\mathrm{LU}} G_{\mathrm{LU}} \zeta_{\mathrm{U}} W_{\mathrm{U}} \left( \frac{\nu_{\mathrm{LU}}}{4\pi \sigma_{\mathrm{U}}r_{\mathrm{LU}}} \right)^2  \right) \Bigg],
\end{split}
\end{equation}
where $r_{\mathrm{LU}}$ is defined in (\ref{app_r_LU}).

\par
By taking the expectation over the distribution of $(\theta_{\mathrm{LU}}$ and $W_{\mathrm{U}}$, the average rate is rewritten as
\begin{equation}\label{app:R_LU}
\begin{split}
& \mathcal{R}_{\mathrm{LU}} = \int_{0}^{\theta_{\mathrm{max}}} \int_0^{\infty} B_{\mathrm{LU}} f_{W_{\mathrm{U}}}(w) f_{\theta_{\mathrm{LU}}}(\theta) \\
& \times \log_2\left(1 + \rho_{\mathrm{LU}} G_{\mathrm{LU}} \zeta_{\mathrm{U}} w \left( \frac{\nu_{\mathrm{LU}}}{4\pi \sigma_{\mathrm{U}}r_{\mathrm{LU}}(\theta)} \right)^2  \right) \mathrm{d}w \mathrm{d}\theta.
\end{split}
\end{equation}

\section{Proof of Theorem~\ref{theorem4}}\label{app:theorem4}
In this appendix, we start with the data rate of the leader to the $i$-th follower link
\begin{equation}
\begin{split}
& \mathcal{R}_{\mathrm{LF},i}
= B_{\mathrm{LF}} \\
& \times \log_2\left(1 + \rho_{\mathrm{LF}} G_{\mathrm{LF}} \zeta_{\mathrm{F}} W_{\mathrm{F}} \left( \frac{\nu_{\mathrm{LF}}}{4\pi \sigma_{\mathrm{F}}r_{\mathrm{LF}}} \right)^2\right),
\end{split}
\end{equation}
where $r_{\mathrm{LF}}$ denotes the Euclidean distance between the leader satellite and the $i$-th follower,
\begin{equation}
r_{\mathrm{LF}}\left(\psi\right) = R_{\mathrm{sat}} \sqrt{2 (1 - \cos \psi)}.
\end{equation}

\par
Similarly, the data rate from the $i$-th follower satellite to the user is given by
\begin{equation}
\begin{split}
& \mathcal{R}_{\mathrm{FU},i} = B_{\mathrm{FU}} \\
& \times \log_2\left(1 + \rho_{\mathrm{FU}} G_{\mathrm{FU}} \zeta_{\mathrm{U}} W_{\mathrm{U}} \left( \frac{\nu_{\mathrm{FU}}}{4\pi \sigma_{\mathrm{U}}r_{\mathrm{FU}}} \right)^2 \right),
\end{split}
\end{equation}
where $r_{\mathrm{FU}}$ is given in (\ref{app:r_FU}).

\par
Since the links are statistically identical and independent, the summation in~\eqref{def:R_c} can be rewritten as the product of the number of follower satellites and the expected value of the minimum rate per link:
\begin{equation}
    \mathcal{R}_{\mathrm{LFs}} = \mathcal{R}_{\mathrm{LU}}+N_{\mathrm{F}} \, \mathbb{E}_{W_{\mathrm{U}},r_\mathrm{LF},r_\mathrm{FU}} \left[ \min\left\{ \mathcal{R}_{\mathrm{LF},i}, \mathcal{R}_{\mathrm{FU},i} \right\} \right],
\end{equation}
where $\mathcal{R}_{\mathrm{LU}}$ is given in (\ref{app:R_LU}). By integrating over the distribution of the four random variables, the average rate of the satellite cluster can be expressed as
\begin{equation}
\begin{split}
\mathcal{R}_{\mathrm{LFs}}=N_{\mathrm{F}}  \int_{0}^{\infty}E(w)f_{W_{\mathrm{U}}}\left(w\right) \mathrm{d}w+\mathcal{R}_{\mathrm{LU}},
\end{split}
\end{equation}
where $E\left(w\right)$ is the average data rate with a given small-scale fading $u$
\begin{equation}
\begin{split}
& E\left(w\right) = \int_{0}^{2\pi} \int_{0}^{ \theta_{\mathrm{cap}}} \int_{0}^{\theta_{\mathrm{max}}} f_{\theta_{\mathrm{LU}}}(\theta) f_{\psi}(\psi) f_{\varphi}(\varphi) \\
& \times \min\left\{ \mathcal{R}_{\mathrm{LF},i}(\psi), \mathcal{R}_{\mathrm{FU},i}(w,\theta,\psi,\varphi) \right\} \mathrm{d}\theta \mathrm{d}\psi \mathrm{d}\varphi.
\end{split}
\end{equation}
By substituting the PDFs of $\psi$ and $\varphi$ into the above expression, the final result can be obtained.

% \begin{equation}
% \begin{split}
%     \bar{R}_{\mathrm{LFs}} 
%     &= N_{\mathrm{F}}\int_{\theta}^{\theta_{\mathrm{cap}}} \int_{0}^{\theta_{\mathrm{max}}} \int_0^{\infty}\int_0^{\infty}
%     R_{\mathrm{F}}\left(W_{\mathrm{FU}},W_{\mathrm{LF}},\theta_{\mathrm{FU}},\theta_{\mathrm{F}}\right) \\
%     &\quad \times f_{W_s}(W_{\mathrm{FU}})f_{W_s}(W_{\mathrm{LF}}) f_{\theta_{\mathrm{FU}}}(\theta_{\mathrm{FU}}) f_{\theta_{\mathrm{F}}}(\theta_{\mathrm{F}}) \\
%     &\quad \times dW_{\mathrm{FU}}dW_{\mathrm{LF}} d\theta_{\mathrm{FU}}d\theta_{\mathrm{F}} ,
% \end{split}
% \end{equation}
% where
% \begin{equation}\label{eq:RF}
% \begin{split}
% &R_{\mathrm{F}}\left(W_{\mathrm{FU}},W_{\mathrm{LF}},\theta_{\mathrm{LU}},\theta_{\mathrm{F}},\phi_{\mathrm{F}}\right)\\
% &=B \log_2\left(1 + C\cdot \min \left\{\frac{W_{\mathrm{LF}}}{r_{\mathrm{LF}}^2},\frac{W_{\mathrm{FU}}}{r_{\mathrm{FU}}^2}\right\}\right)
%     \end{split}
% \end{equation}
% \begin{equation}
% r_{\mathrm{LF}}\left(\theta_{\mathrm{F}}\right) = R_{\mathrm{sat}} \sqrt{2 (1 - \cos \theta_{\mathrm{F}})}.
% \end{equation}

\bibliographystyle{IEEEtran}
\bibliography{reference}

\begin{thebibliography}{10}
\providecommand{\url}[1]{#1}
\csname url@samestyle\endcsname
\providecommand{\newblock}{\relax}
\providecommand{\bibinfo}[2]{#2}
\providecommand{\BIBentrySTDinterwordspacing}{\spaceskip=0pt\relax}
\providecommand{\BIBentryALTinterwordstretchfactor}{4}
\providecommand{\BIBentryALTinterwordspacing}{\spaceskip=\fontdimen2\font plus
\BIBentryALTinterwordstretchfactor\fontdimen3\font minus \fontdimen4\font\relax}
\providecommand{\BIBforeignlanguage}[2]{{%
\expandafter\ifx\csname l@#1\endcsname\relax
\typeout{** WARNING: IEEEtran.bst: No hyphenation pattern has been}%
\typeout{** loaded for the language `#1'. Using the pattern for}%
\typeout{** the default language instead.}%
\else
\language=\csname l@#1\endcsname
\fi
#2}}
\providecommand{\BIBdecl}{\relax}
\BIBdecl

\bibitem{xiao2022leo}
Z.~Xiao, J.~Yang, T.~Mao, C.~Xu, R.~Zhang, Z.~Han, and X.-G. Xia, ``{LEO satellite access network (LEO-SAN) toward 6G: Challenges and approaches},'' \emph{IEEE Wireless Communications}, vol.~31, no.~2, pp. 89--96, 2022.

\bibitem{wang2022stochastic}
R.~Wang, M.~A. Kishk, and M.-S. Alouini, ``Stochastic geometry-based low latency routing in massive {LEO} satellite networks,'' \emph{IEEE Transactions on Aerospace and Electronic Systems}, vol.~58, no.~5, pp. 3881--3894, 2022.

\bibitem{dwivedi2023performance}
A.~K. Dwivedi, S.~Chaudhari, N.~Varshney, and P.~K. Varshney, ``Performance analysis of {LEO} satellite-based {IoT} networks in the presence of interference,'' \emph{IEEE Internet of Things Journal}, vol.~11, no.~5, pp. 8783--8799, 2023.

\bibitem{campbell2003planning}
M.~E. Campbell, ``Planning algorithm for multiple satellite clusters,'' \emph{Journal of Guidance, Control, and Dynamics}, vol.~26, no.~5, pp. 770--780, 2003.

\bibitem{liu2018survey}
G.-P. Liu and S.~Zhang, ``A survey on formation control of small satellites,'' \emph{Proceedings of the IEEE}, vol. 106, no.~3, pp. 440--457, 2018.

\bibitem{goh2019leader}
S.~T. Goh, K.-S. Low, and E.-K. Poh, ``Leader-followers satellite formation control for low-thrust small satellite application,'' in \emph{International Symposium on Space Technology and Science (ISTS)}, 2019.

\bibitem{ahn2012leader}
H.-S. Ahn, ``Leader--follower type relative position keeping in satellite formation flying via robust exponential stabilization,'' \emph{International Journal of Robust and Nonlinear Control}, vol.~22, no.~18, pp. 2084--2099, 2012.

\bibitem{wang2025modeling}
R.~Wang, M.~A. Kishk, and M.-S. Alouini, ``Modeling and analysis of non-terrestrial networks by spherical stochastic geometry: A survey,'' \emph{IEEE Communications Surveys \& Tutorials}, 2025, early Access.

\bibitem{lee2024analyzing}
M.~Lee, S.~Kim, M.~Kim, D.-H. Jung, and J.~Choi, ``Analyzing downlink coverage in clustered low {E}arth orbit satellite constellations: {A} stochastic geometry approach,'' 2024, available online: https://arxiv.org/abs/2402.16307.

\bibitem{eyer2007formation}
J.~K. Eyer, C.~J. Damaren, R.~E. Zee, and E.~Cannon, ``A formation flying control algorithm for the canx-45 low {E}arth orbit nanosatellite mission,'' \emph{Space Technology}, vol.~27, no.~4, p. 147, 2007.

\bibitem{popov2021development}
A.~M. Popov, I.~Kostin, J.~Fadeeva, and B.~Andrievsky, ``Development and simulation of motion control system for small satellites formation,'' \emph{Electronics}, vol.~10, no.~24, p. 3111, 2021.

\bibitem{yu2016virtual}
Q.-Y. Yu, W.-X. Meng, M.-C. Yang, L.-M. Zheng, and Z.-Z. Zhang, ``Virtual multi-beamforming for distributed satellite clusters in space information networks,'' \emph{IEEE Wireless Communications}, vol.~23, no.~1, pp. 95--101, 2016.

\bibitem{wang2024ultra}
R.~Wang, M.~A. Kishk, and M.-S. Alouini, ``Ultra reliable low latency routing in {LEO} satellite constellations: A stochastic geometry approach,'' \emph{IEEE Journal on Selected Areas in Communications}, vol.~42, no.~5, pp. 1231--1245, 2024.

\bibitem{jung2023satellite}
D.-H. Jung, J.-G. Ryu, and J.~Choi, ``Satellite clustering for non-terrestrial networks: Orbital configuration-dependent outage analysis,'' \emph{IEEE Wireless Communications Letters}, vol.~13, no.~2, pp. 550--554, 2023.

\bibitem{jung2023satellite2}
D.-H. Jung, G.~Im, J.-G. Ryu, S.~Park, H.~Yu, and J.~Choi, ``Satellite clustering for non-terrestrial networks: Concept, architectures, and applications,'' \emph{IEEE Vehicular Technology Magazine}, vol.~18, no.~3, pp. 29--37, 2023.

\bibitem{al2020modeling}
A.~Al-Hourani and I.~Guvenc, ``On modeling satellite-to-ground path-loss in urban environments,'' \emph{IEEE Communications Letters}, vol.~25, no.~3, pp. 696--700, 2020.

\bibitem{chae2023performance}
S.~H. Chae, H.~Lim, H.~Lee, and B.~C. Jung, ``Performance analysis of dense low {E}arth orbit satellite communication networks with stochastic geometry,'' \emph{Journal of Communications and Networks}, vol.~25, no.~2, pp. 208--221, 2023.

\bibitem{wang2025satellite}
R.~Wang, M.~A. Kishk, H.~H. Yang, and M.-S. Alouini, ``Satellite-terrestrial routing or inter-satellite routing? a stochastic geometry perspective,'' \emph{IEEE Transactions on Aerospace and Electronic Systems}, pp. 1--17, 2025, early Access.

\bibitem{chen20243}
Q.~Chen, L.~Yang, Y.~Zhao, Y.~Wang, H.~Zhou, and X.~Chen, ``3-{ISL} topology: Routing properties and performance in {LEO} mega-constellation networks,'' \emph{IEEE Transactions on Aerospace and Electronic Systems}, vol.~61, no.~2, pp. 4961--4972, 2025.

\bibitem{wang2022conditional}
R.~Wang, A.~Talgat, M.~A. Kishk, and M.-S. Alouini, ``Conditional contact angle distribution in {LEO} satellite-relayed transmission,'' \emph{IEEE Communications Letters}, vol.~26, no.~11, pp. 2735--2739, 2022.

\bibitem{talgat2024maximizing}
A.~Talgat, M.~A. Kishk, and M.-S. Alouini, ``Maximizing uplink data transmission of {LEO}-satellite-based wireless-powered {IoT},'' \emph{IEEE Internet of Things Journal}, vol.~11, no.~17, pp. 28\,975--28\,987, 2024.

\bibitem{wang2022ultra}
R.~Wang, M.~A. Kishk, and M.-S. Alouini, ``Ultra-dense {LEO} satellite-based communication systems: A novel modeling technique,'' \emph{IEEE Communications Magazine}, vol.~60, no.~4, pp. 25--31, 2022.

\bibitem{talgat2020stochastic}
A.~Talgat, M.~A. Kishk, and M.-S. Alouini, ``Stochastic geometry-based analysis of {LEO} satellite communication systems,'' \emph{IEEE Communications Letters}, vol.~25, no.~8, pp. 2458--2462, 2020.

\bibitem{huang2021uplink}
Q.~Huang, M.~Lin, W.-P. Zhu, J.~Cheng, and M.-S. Alouini, ``Uplink massive access in mixed {RF/FSO} satellite-aerial-terrestrial networks,'' \emph{IEEE Transactions on Communications}, vol.~69, no.~4, pp. 2413--2426, 2021.

\bibitem{jung2022performance}
D.-H. Jung, J.-G. Ryu, W.-J. Byun, and J.~Choi, ``Performance analysis of satellite communication system under the shadowed-{R}ician fading: A stochastic geometry approach,'' \emph{IEEE Transactions on Communications}, vol.~70, no.~4, pp. 2707--2721, 2022.

\bibitem{jia2021uplink}
H.~Jia, Z.~Ni, C.~Jiang, L.~Kuang, and J.~Lu, ``Uplink interference and performance analysis for megasatellite constellation,'' \emph{IEEE Internet of Things Journal}, vol.~9, no.~6, pp. 4318--4329, 2021.

\bibitem{al2021analytic}
A.~Al-Hourani, ``An analytic approach for modeling the coverage performance of dense satellite networks,'' \emph{IEEE Wireless Communications Letters}, vol.~10, no.~4, pp. 897--901, 2021.

\end{thebibliography}

\end{document}